\documentclass[a4paper,USenglish,cleveref]{lipics-v2021}
\pdfoutput=1
\nolinenumbers
\ccsdesc[300]{Theory of computation}  \keywords{Cartesian-tree matching, extended Burrows--Wheeler transform, construction algorithm, generalized pattern matching} \category{} \title{Extending the Burrows--Wheeler Transform for Cartesian Tree Matching and Constructing It}
	\author{Eric M. Osterkamp}{University of M\"{u}nster, Germany}{osterkamp@uni-muenster.de}{https://orcid.org/0009-0007-0368-5008}{}
	\author{Dominik K\"{o}ppl}{University of Yamanashi, Japan \and \url{https://dkppl.de/}}{dkppl@yamanashi.ac.jp}{https://orcid.org/0000-0002-8721-4444}{JSPS KAKENHI Grant Number 23H04378 and Yamanashi Wakate Grant Number 2291}
\authorrunning{E.\ M.\ Osterkamp and D.\ K\"{o}ppl}
\Copyright{Eric M. Osterkamp and Dominik Köppl}
\hideLIPIcs

\usepackage[utf8]{inputenc}
\usepackage{xcolor}
\definecolor{seabornBlue}{RGB}{76,114,176}
\definecolor{seabornGreen}{RGB}{85,168,104}
\definecolor{seabornRed}{RGB}{196,78,82}

\usepackage{graphicx}
\usepackage{tikz}
\usepackage{forest}
\usepackage{subcaption}
\usepackage[linesnumbered,ruled,vlined]{algorithm2e}

\SetCommentSty{mycommfont}
\usepackage[labelfont=bf]{caption}
\usepackage{lscape}

\usepackage[bb=boondox]{mathalfa}
\usepackage{amsmath,amssymb,amsthm,mathtools,enumerate,enumitem}
\usepackage{thmtools}
\usepackage{thm-restate}
\usepackage{stmaryrd}

\usepackage{scalerel}
\newcommand\Infty{\hstretch{0.5}{\infty}}
\newcommand{\ctmatch}{\approx}
\newcommand{\cteq}{=_{\omega}}
\newcommand{\ctprec}{\prec_{\omega}}
\newcommand{\ctpeq}{\preceq_{\omega}}

\newcommand{\mpdenc}[1]{\langle #1 \rangle}
\newcommand{\rpdenc}[1]{\langle #1 \rangle_\textup{r}}

\newcommand{\rtsenc}[1]{\llbracket #1 \rrbracket_{}}

\newcommand{\lcpinfty}[1]{\textup{LCP}_{#1}^\infty}
\newcommand{\ftarr}[1]{\textup{F}_{#1}}
\newcommand{\ltarr}[1]{\textup{L}_{#1}}

\newcommand{\carr}[1]{\textup{CA}_{#1}}
\newcommand{\icarr}[1]{\textup{CA}_{#1}^{-1}}

\newcommand{\alltexts}[2]{#1_{1},...,#1_{#2}}

\newcommand{\extsigma}{\Sigma_{\$}}

\newcommand{\rank}[3]{\textup{rank}_{#1}(#2,#3)}
\newcommand{\select}[3]{\textup{select}_{#1}(#2,#3)}
\newcommand{\rangecount}[5]{\textup{rnkcnt}_{#1}(#2,#3,#4,#5)}
\newcommand{\maxi}[3]{\textup{MI}_{#1}(#2,#3)}
\newcommand{\lcplength}[2]{\textup{lcp}(#1,#2)}
\newcommand{\lcpcount}[2]{\textup{lcp}^\infty(#1,#2)}
\newcommand{\inst}[3]{\textup{insert}_{#1}(#2,#3)}
\newcommand{\dlt}[2]{\textup{delete}_{#1}(#2)}
\newcommand{\flarr}[1]{\textup{FL}_{#1}}
\newcommand{\flmap}[2]{\textup{FL}_{#1}[#2]}
\newcommand{\lfarr}[1]{\textup{LF}_{#1}}
\newcommand{\lfmap}[2]{\textup{LF}_{#1}[#2]}

\newcommand{\ctree}[1]{\textup{ct}(#1)}

\newcommand{\crange}[2]{\textup{CR}_{#1}(#2)}

\newcommand{\rnv}[4]{\textup{RNV}_{#1}(#2,#3,#4)}
\newcommand{\rta}[2]{\textup{Rot}(#1,#2)}
\newcommand{\conj}[2]{C_{#1}({#2})}

\newcommand{\wurz}[1]{\textup{root}(#1)}
\newcommand{\epo}[1]{\textup{exp}(#1)}
\newcommand{\prev}[2]{\textup{prev}_{#1}(#2)}
\newcommand{\TT}{\mathcal{T}}
\newcommand{\ST}{\mathcal{S}}
\newcommand{\RT}{\mathcal{R}}

\newcommand{\cnt}[2]{\textup{cnt}_{#1}(#2)}
\newcommand{\plcp}[2]{\textup{plcp}^\infty_{#1}(#2)}
\newcommand{\slcp}[2]{\textup{slcp}^\infty_{#1}(#2)}

\newcommand*{\Count}{\textsc{Count}}

\newtheorem{thm}{Theorem}[section]
\newtheorem{problem}{Problem}

\newtheorem{cor}[thm]{Corollary}
\newtheorem{lem}[thm]{Lemma}

\theoremstyle{remark}
\newtheorem{rem}[thm]{Remark}

\theoremstyle{definition}

\newenvironment{rdir}{\noindent\textbf{($\Rightarrow$).~}}{}
\newenvironment{ldir}{\noindent\textbf{($\Leftarrow$).~}}{}
\newenvironment{cs1}{\noindent\textbf{Case 1.~}}{}
\newenvironment{cs2}{\noindent\textbf{Case 2.~}}{}
\newenvironment{cs21}{\noindent\textbf{Case 2.1.~}}{}
\newenvironment{cs22}{\noindent\textbf{Case 2.2.~}}{}
\newenvironment{cs23}{\noindent\textbf{Case 2.3.~}}{}
\newenvironment{cs3}{\noindent\textbf{Case 3.~}}{}

\newcommand*{\Block}[1]{\noindent \textbf{#1}\ }
\begin{document}
	\maketitle
	\begin{abstract}
		Cartesian tree matching is a form of generalized pattern matching where a substring of the text matches with the pattern if they share the same Cartesian tree.
		This form of matching finds application for time series of stock prices and can be of interest for melody matching between musical scores.
		For the indexing problem, the state-of-the-art data structure is a Burrows--Wheeler transform based solution due to [Kim and Cho, CPM'21], which uses nearly succinct space and can 
		count the number of substrings that Cartesian tree match with a pattern in time linear in the pattern length.
		The authors address the construction of their data structure with a straight-forward solution that, however, requires pointer-based data structures,
		which asymptotically need more space than compact solutions [Kim and Cho, CPM'21, Section~A.4].
		We address this bottleneck by a construction that requires compact space and has a time complexity linear in the product of the text length with some logarithmic terms.
		Additionally, we can extend this index for indexing multiple circular texts in the spirit of the extended Burrows--Wheeler transform without sacrificing the time and space complexities.
		We present this index in a dynamic variant, where we pay a logarithmic slowdown and need compact space for the extra functionality that we can incrementally add texts.
		Our extended setting is of interest for finding repetitive motifs common in the aforementioned applications, independent of offsets and scaling.
	\end{abstract}

	\section{Introduction}\label{section:intro}

String matching is ubiquitous in computer science, and its variations are custom-made to solve a wide variety of problems.
We here focus on a special kind of variation called \emph{substring consistent equivalence relation} (\emph{SCER})~\cite{matsuoka16generalized}.
Two strings $X$ and $Y$ are said to SCER-match if they have the same length and all substrings of equal length starting at the same text position SCER-match, i.e.,
$X[i..j]$ SCER-matches with $Y[i..j]$ for all $1 \le i \le j \le |X| = |Y|$.
A specific instance of SCER-matching is \emph{order-preserving matching}~\cite{kubica13orderpreserving,kim14orderpreserving}, which has been studied for the analysis of numerical time series.
The aim of order-preserving matching is to match two strings if the relative order of the symbols in both strings is the same.
Order-preserving matching therefore can find matches independently of value offsets and scaling.

Since order-preserving matching takes the global order of the symbols in a string into account, it may be too strict in applications that 
primarily consider the local ordering of ranges partitioned by the peaks.
In fact, time series of stock prices are such a case, where a common pattern called the \emph{head-and-shoulder}~\cite{fu07stock} involves one absolute peak (head) and neighboring local peaks (shoulders),
where each of these neighboring peaks can be treated individually to match similar head-and-shoulder patterns with slightly changed local peaks.
For such a reason, Park et al.~\cite{park20patternsperiods} proposed \emph{Cartesian tree matching}.
Cartesian tree matching relaxes the notion of order-preserving matching in the sense that an order-preserving match is always a Cartesian tree match, but not necessarily the other way around.
For instance, the two strings of digits $1647253$ and $2537164$ fit into the head-and-shoulder pattern, do not order-preserving match, but Cartesian tree match.
For that, Cartesian tree matching compares the Cartesian trees of two strings to decide whether they match.
The \emph{Cartesian tree}~\cite{vuillemin80unifying} is a binary tree built upon an array of numbers.
If all numbers are distinct, it is the min-heap whose in-order traversal retrieves the original array.
Since its inception, Cartesian tree matching and variations thereof have attracted interest from researchers,
whose studies include multiple pattern matching~\cite{park20patternsperiods,song21fastmatchingalg}, approximate matching~\cite{auvray23approximateswaps,kim24approximate}, substring matching~\cite{faro22ctmsubstring}, subsequence matching~\cite{oizumi22subsequence}, indeterminate matching~\cite{gawrychowski20indeterminate}, cover matching~\cite{kikuchi20covers}, and palindromic matching~\cite{funakoshi24computing}.

In addition to stock prices, applications of generalized pattern matching can also be found in melody matching between musical scores. 
Pattern matching music scores such as differences, directions, or magnitudes have been studied~\cite{hiraga97structural}.
However, the detection of repetitions in a piece of music has also been considered of importance~\cite{foster12method}.
The question therefore is whether we can find repetitive substrings that Cartesian tree match with a repetition of a set of melodic motifs (i.e., input texts).
For that to be efficient, we want to index these motifs.

An index for Cartesian tree matching of a text string $T$ is a data structure built upon $T$ that can report the number of substrings of $T$ that Cartesian tree match a given pattern.
Given $T$ has $n$ symbols drawn from an integer alphabet of size~$\sigma$,
Park et al.~\cite{park20patternsperiods} and Nishimoto et al.~\cite{nishimoto21heaps} proposed indexes for Cartesian tree matching of $T$ that both occupy $O(n \lg n)$ bits of space and are constructed with $O(n \lg n)$ and $O(n \sigma \lg n)$ additional bits of space, respectively. 
Park et al.~\cite[Section~5.1]{park20patternsperiods} achieve $O(m \lg \sigma)$ time and
Nishimoto et al.~\cite[Section~5.1]{nishimoto21heaps} $O(m (\sigma + \lg m) + occ)$ time for answering count queries, where $occ$ is the number of occurrences of the pattern of length~$m$.

Adapting Ferragina and Mantaci's FM-index for exact string matching~\cite{ferragina00fmindex}, Kim and Cho~\cite{kim21compact} proposed an index that occupies $3n + o(n)$ bits of space,
and answers a count query in $O(m)$ time for a pattern of length~$m$.
For construction, they proposed a straight-forward solution that takes as input the Cartesian suffix tree of Park et al.~\cite{park20patternsperiods},
which, however, requires $O(n \lg n)$ additional bits of space.

We address two goals in this paper. The first is a construction algorithm for Kim and Cho's index that takes compact working space.
Second, while all aforementioned indexes can partially address the problem by indexing multiple texts, 
it is hard to detect whether the pattern is a repetition of one of the input texts, which is of interest in case of indexing melodic motifs that can repeat.
Here, our goal is an index that can find such matches, even if they start with different offsets of the same repetition.
In concrete words, our aim is to index multiple texts for Cartesian pattern matching.
The search space for a pattern are the texts that are considered to be infinite concatenations with themselves. 

In this paper, we propose an extension of the index of Kim and Cho~\cite{kim21compact}
for Cartesian tree matching with techniques of the extended Burrows--Wheeler transform~\cite{mantaci07ebwt},
and call the resulting data structure the cBWT index.
We show that we can compute the cBWT index in $O(n \frac{\lg \sigma \lg n}{\lg \lg n})$ time and $O(n \lg \sigma)$ bits of space, where $n$ is the total length of all texts to index.
We can also compute the original index of Kim and Cho within the same complexities.
Our ideas stem from Hashimoto et al.'s~\cite{hashimoto22computing} and Iseri et al.'s~\cite{iseri24breaking} construction algorithms of indexes for parameterized matching, 
which we recently extended for multiple circular texts~\cite{osterkamp24extending}.
	During construction, the cBWT index supports the backward search and count queries for pattern strings in $O(m \frac{\lg \sigma \lg n}{\lg\lg n})$ time, where $m$ is the pattern length.

	\section{Preliminaries}\label{section:preliminaries}
	
	Let $\lg = \log_2$ denote the logarithm to base two.
	We assume a random access model with word size $\Omega(\lg n)$, where $n$ denotes the input size.
	An interval $\{i,i+1,..,j-1,j\}$ of integers is denoted by $[i..j]$, where $[i..j] = \emptyset$ if $j < i$.
	
\Block{Strings.}
	Let $\Sigma$ denote an \emph{alphabet}.
	We call elements of $\Sigma$ \emph{symbols}, a sequence of symbols from $\Sigma$ a \emph{string over $\Sigma$}, and denote the set of strings over $\Sigma$ by $\Sigma^*$.
	Let $U,V,W\in \Sigma^*$.
	The \emph{concatenation} of $U$ and $V$ is denoted by $U\cdot V$ or $UV$.
	We write $U^k$ if we concatenate $k$ instances of $U$ for a non-negative integer $k$, and $U^\omega$ for the string obtained by infinitely concatenating $U$.
	We call $U$ \emph{primitive} if $U = X^k$ implies $U = X$ and $k = 1$.
	It is known that for every $U \in \Sigma^*$ there exists a unique primitive $X \in \Sigma^*$ and a unique integer $k$ such that $U = X^k$, denoted by $\wurz{U}$ and $\epo{U}$, respectively.
	If $X = UVW$, then $U$ is called a \emph{prefix}, $W$ a \emph{suffix}, and $U,V,W$ \emph{substrings} of $X$.
	The \emph{length} $\left|U\right|$ of $U$ is the number of symbols in $U$, $\lcplength{U}{V}$ reports the length of the longest common prefix of $U$ and $V$, $\varepsilon$ denotes the unique string of length $0$, and we define $\Sigma^+ = \Sigma^*-\{\varepsilon\}$.
	Let $X \in \Sigma^+$ and $i,j \in [1..\left|X\right|]$.
	Then $X[i]$ denotes the $i$-th symbol in $X$, and $X[i..j] = X[i]\cdots X[j]$, where $X[i..j] = \varepsilon$ if $j < i$, $X[..i] = X[1..i]$, and $X[i..] = X[i..\left|V\right|]$.
	For notational convenience, $X[..k] = \varepsilon$ if $k \leq 0$, and $V[k..] = \varepsilon$ if $k \geq \left|V\right| + 1$.
	Let $\rta{X}{0} = X$ and $\rta{X}{k+1} = \rta{X}{k}[2..]\cdot\rta{X}{k}[1]$ for each non-negative integer $k$, i.e., $\rta{X}{k}$ denotes the $k$-th \emph{left rotation} of $X$.
	The left rotations of $X$ for $k \in [0..\left|X\right|-1]$ are the \emph{conjugates} of $X$.
	We write $U < V$ if and only if (a) $U = V[..\left|U\right|]$ and $\left|U\right| < \left|V\right|$ or (b) $U[\lcplength{U}{V} +1] < V[\lcplength{U}{V} +1]$.

\Block{Queries on Strings.}
	Let $V \in \Sigma^+$ be a (dynamic) string, $c \leq d \in \Sigma$, $i \in [1..\left|V\right|+1]$, $j \in [0..\left|V\right|]$ and $k \in [1..\left|V\right|]$.
	Then $\inst{V}{i}{c}$ inserts the symbol $c$ at position $i$ of $V$, $\dlt{V}{k}$ deletes the $k$-th entry of $V$, $\rank{c}{V}{j}$ returns the number of occurrences of $c$ in $V[..j]$, $\rangecount{V}{i}{j}{c}{d}$ returns $\left|\{x \in [i..j] \mid c\leq V[x] \leq d\}\right|$, $\select{c}{V}{i}$ returns the index of the $i$-th occurrence of $c$ in $V$ if $i \leq \rank{c}{V}{\left|V\right|}$, $\rnv{V}{i}{j}{c}$ returns the smallest value in $V[i..j]$ larger than $c$ if it exists, and $\maxi{V}{j}{c}$ returns the maximal interval $[\ell..r]$ such that $0 \leq \ell \leq j \leq r \leq \left|V\right|$ and $V[x] \geq c$ for every $x \in [\ell+1..r]$.
\begin{figure}[tbp]
		\centering
		\begin{tabular}{|c|c|c|c|c|c|c|c|c|}
			\hline
			$i$&1&2&3&4&5&6&7&8\\
			\hline
			$V[i]$&$\mathtt{7}$&$\mathtt{14}$&$\mathtt{5}$&$\mathtt{1}$&$\mathtt{11}$&$\mathtt{27}$&$\mathtt{11}$&$\mathtt{7}$\\
			\hline
		\end{tabular}
		\caption{Example integer string to illustrate queries. Here, $\rank{11}{V}{5}= 1$, $\rangecount{V}{2}{7}{7}{14} = 4$, $\select{11}{V}{2}= 7$, $\rnv{V}{1}{5}{8} = 11$, and $\maxi{V}{5}{6} = [4..8]$.}
		\label{fig:queries}
	\end{figure}
	See Figure~\ref{fig:queries} for examples.
	We represent strings by the following dynamic data structure, which supports the aforementioned operations and queries.
	
	\begin{lem}[{\cite[Lemma~4]{iseri24breaking}}]\label{lemma:dynamicstringrmq}
		A dynamic string of length $n$ over $[0..\sigma]$ with $\sigma \leq n^{O(1)}$ can be stored in a data structure occupying $O(n \lg \sigma)$ bits of space, supporting insertion, deletion and the queries \textup{access}, \textup{rank}, \textup{rnkcnt}, \textup{select}, \textup{RNV} and \textup{MI} in $O(\frac{\lg \sigma \lg n}{\lg \lg n})$ time.
	\end{lem}
	
	\Block{Alphabet.}
	Throughout, we will work with the integer alphabet $\Sigma = [0..\sigma]$, where $\sigma \leq n^{O(1)}$, and a special symbol $\$ \not\in \Sigma$ stipulated to be smaller than any symbol from $\Sigma$.
	The special symbol is motivated by the construction algorithm, and as a delimiter when a string should not be considered circular, as in the index of Kim and Cho~\cite{kim21compact}.
	Let $\extsigma = \Sigma \cup \{\$\}$.
	
	\Block{Cartesian Tree Matching.}
	The \emph{Cartesian tree} $\ctree{V}$ of a string $V \in \extsigma^*$ is a binary tree defined as follows. 
	If $V = \varepsilon$, then $\ctree{V}$ is the empty tree. 
	If $V \neq \varepsilon$, let $i$ denote the position of the smallest symbol in $V$, where ties are broken with respect to text position. 
	Then $\ctree{V}$ has $V[i]$ as its root, $\ctree{V[..i-1]}$ as its left subtree and $\ctree{V[i+1..]}$ as its right subtree.
	We say that two strings $U,V \in \extsigma^*$ \emph{Cartesian tree match (ct-match)} if and only if $\ctree{U} = \ctree{V}$, and write $U \ctmatch V$.
	For instance, in \cref{fig:cartesian_tree} the substring $T[3..6] = \mathtt{5178}$ of $T = \mathtt{625178265}$ ct-matches $P = \mathtt{7347}$ while $T[4..7] =\mathtt{1782}$ does not.

	\begin{figure}
		\centering
		\begin{subfigure}{.35\textwidth}
			\begin{forest}
				[1
				[2
				[6
				[,phantom]
				[,phantom]
				]
				[5
				[,phantom]
				[,phantom]
				]
				]
				[2
				[7
				[,phantom]
				[8]
				]
				[5
				[6]
				[,phantom]
				]
				]
				]
			\end{forest}
		\caption{$\ctree{T}$.
		}
		\end{subfigure}\begin{subfigure}{.2\textwidth}
			\begin{forest}
				[3
				[7
				[,phantom]
				[,phantom]
				]
				[4
				[,phantom]
				[7]
				]
				]
			\end{forest}
			\caption{$\ctree{P}$.
			}
		\end{subfigure}\begin{subfigure}{.2\textwidth}
			\begin{forest}
				[1
				[5
				[,phantom]
				[,phantom]
				]
				[7
				[,phantom]
				[8]
				]
				]
			\end{forest}
		\caption{\ctree{$T[3..6]$}.
		}
		\end{subfigure}\begin{subfigure}{.2\textwidth}
			\begin{forest}
				[1
				[,phantom
				[,phantom]
				[,phantom
				[,phantom]
				[,phantom]
				]
				]
				[2
				[7
				[,phantom]
				[8]
				]
				[,phantom
				[,phantom]
				[,phantom]
				]
				]
				]
			\end{forest}
		\caption{\ctree{$T[4..7]$}.
		}
		\end{subfigure}\caption{Cartesian trees of $T = \mathtt{625178265}$, $P = \mathtt{7347}$, $T[3..6] = \mathtt{5178}$, and $T[4..7] = \mathtt{1782}$.
		Here, $P \ctmatch T[3..6]$ and $P\not\ctmatch T[4..7]$.
	}
	\label{fig:cartesian_tree}
	\end{figure}
	\Block{Parent Distance Encoding.}
	We give a variant of Park et al.'s~\cite{park20patternsperiods} encoding scheme for representing Cartesian trees, which reduces the computation of a ct-match to checking if the encoded strings exactly match.
	Let $\Infty \not\in \extsigma$ denote a symbol larger than any integer.
	The \emph{parent distance encoding} $\mpdenc{V}$ of $V \in \extsigma$ is a string of length $\left|V\right|$ over $\extsigma \cup \{\Infty\}$ such that 
	\[
	\mpdenc{V}[i] =
	\begin{cases*}
		\Infty & \text{if $\$ \neq V[i] < \min\{V[j] \mid j \in [..i-1]\}$,} \\
		V[i] & \text{if $V[i] = \$$, }\\
		i - \max \{j \in [.. i-1] \mid V[j] < V[i]\} & \text{otherwise,}
	\end{cases*}
	\]
	for each $i \in [1..\left|V\right|]$.
	Note that $\mpdenc{V}[..i] = \mpdenc{V[..i]}$ for each $V\in\extsigma^+$ and $i \in [1..\left|V\right|]$.
	For example, $\mpdenc{\mathtt{41327\$3}} = \mathtt{\Infty \Infty 121\$1}$.
	
	\begin{lem}[{\cite[Theorem~1]{park20patternsperiods}}]\label{lem:encctmatch}
		Let $U,V \in \Sigma^*$. 
		Then $U \ctmatch V \Leftrightarrow \mpdenc{U} = \mpdenc{V}$. 
	\end{lem}
	\Block{Problem Statement.}
	We are interested in a solution to the following problem.
	\begin{problem}[\Count{}]\label{problem:count}
		Given $\emptyset \neq \TT \subset \Sigma^+$ and $P \in \Sigma^*$, count each of the conjugates of the strings in $\TT$ whose infinite iteration has a prefix ct-matching $P$.
	\end{problem}
	Throughout, let $\emptyset \neq \TT = \{\alltexts{T}{d}\} \subset \Sigma^+$. 
	Our running example consists of the strings $T_1 = \mathtt{512}$, $ T_2 = \mathtt{5363}$ and $T_3 = \mathtt{4478}$ over $\Sigma = [0..8]$.
	Given our running example, the solution to \Count{} for $P_1 = \mathtt{643}$ and $P_2 = \mathtt{5634}$ is $0$ and $2$, respectively.
	Here, $P_2 \ctmatch \rta{T_3}{2}^\omega[..4] \ctmatch \rta{T_1}{2}^\omega[..4]$ since $\mpdenc{P_2} = \mathtt{\Infty1\Infty1} = \mpdenc{\rta{T_3}{2}^\omega[..4]} = \mpdenc{\rta{T_1}{2}^\omega[..4]}$.

	Due to space limitations, the proofs of some technical claims in the following sections are deferred to the appendix.
	
	\section{Compact Index}\label{section:fmindex}	
	
	Let $n = \left|T_1\cdots T_d\right|$, $n_k = \left|T_k\right|$ for each $k \in [1..d]$, and $\conj{\TT}{i} = \rta{T_j}{i - 1 - \sum_{k = 1}^{j-1}n_k}$ for each $i \in [1..n]$, where $j = \min\{h \in [1..d] \mid \sum_{k = 1}^{h}n_k \geq i\}$, i.e., we identify each conjugate of each text $\alltexts{T}{d}$ with its start position inside the concatenation $T_1\cdots T_d$, such that we give them ranks from $1$ to $n$.
In what follows, we put these ranks in a specific order by a permutation of $[1..n]$
	such that the permuted ranks of all conjugates with prefixes of their infinite concatenation ct-matching a pattern form an interval $[\ell..r] \subseteq [1..n]$.
\subsection{Conjugate Array}\label{subsection:conjugatearray}
	We express the permutation of the ranks of all conjugates by the so-called conjugate array, which we will subsequently define.
	To achieve this, we extend the ideas of Mantaci et al.~\cite{mantaci07ebwt} to Cartesian tree matching, and introduce a preorder on $\extsigma^+$.
	For notational convenience, we define the \emph{rotational parent distance encoding}
	$\rpdenc{V}$ of $V \in \extsigma^+$ by $\rpdenc{V} = \mpdenc{V^2}[\left|V\right|+1..]$.

\begin{restatable}{lem}{rotaparent}\label{lem:rotaencodingmatch}
		Let $V, U \in \extsigma^+$. 
		Then $\rpdenc{V} = \rpdenc{U}$ if and only if $V^2 \ctmatch U^2$.
	\end{restatable}
	
	For any $V, U \in \extsigma^+$, let $V \ctpeq U$ if and only if there exists some natural number $i$ such that $\mpdenc{V^\omega[..i]} < \mpdenc{U^\omega[..i]}$ or
	$\wurz{\rpdenc{V}} = \wurz{\rpdenc{U}}$ holds.
	For the defined relation~$\ctpeq$,
	the following result is a consequence of Lemma~\ref{lem:encctmatch} and Lemma~\ref{lem:rotaencodingmatch}.
	
	\begin{cor}\label{cor:preceqtotalorder}
		The relation $\ctpeq$ defines a total preorder on $\extsigma^+$, i.e., the relation $\ctpeq$ is binary, reflexive, transitive and connected.
	\end{cor}
	
	We call this preorder the \emph{$\omega$-preorder}.
	We write $V \cteq U$ if and only if $V \ctpeq U \land U \ctpeq V$, and $V \ctprec U$ if and only if $V \ctpeq U \land V \not\cteq U$.
	For instance, $T_3 = \mathtt{4478} \ctprec \mathtt{125} = \rta{T_1}{1}$ and $T_2 = \mathtt{5363} \cteq \mathtt{6353} = \rta{T_2}{2}$.
	The following results are due to a periodicity argument, and give us a convenient way to compute the $\omega$-preorder of two given strings.
	
	\begin{restatable}{lem}{equality}\label{lem:equality3z}
		Let $V, U \in \extsigma^+$ and $z = \max\{\left|V\right|, \left|U\right|\}$. 
		Then $V \cteq U$ if and only if $\mpdenc{V^{\omega}[..3z]}  = \mpdenc{U^{\omega}[..3z]}$.
	\end{restatable}
	
	\begin{cor}\label{cor:prec3z}
		Let $V, U \in \extsigma^+$ and $z = \max\{\left|V\right|, \left|U\right|\}$.
		Then $V \ctprec U$ if and only if $\mpdenc{V^{\omega}[..3z]}  < \mpdenc{U^{\omega}[..3z]}$.
	\end{cor}
	
	Similarly to Boucher et al.~\cite{boucher21rindexing}, we define the \emph{conjugate array} $\carr{\TT}$ of $\TT$ as the string of length $n$ over $[1..n]$ such that $\carr{T}[i] = j$ if and only if 
	\[
	i - 1 = |\{ k \in [1..n] \mid
	\conj{\TT}{k} \ctprec \conj{\TT}{j}  \text{ or } \conj{\TT}{k} \cteq \conj{\TT}{j} \land k < j\}|,
	\]
	i.e., $i - 1$ is the number of all conjugates smaller than $\conj{\TT}{j}$ according to $\omega$-preorder, where we break ties first with respect to text index, and then with respect to text position.
	By resolving all ties this way, we ensure that $\carr{\TT}$ is well-defined.
	Since $\carr{\TT}$ is a permutation, its inverse, which we denote by $\icarr{\TT}$, is also well-defined. 
	See Figure~\ref{fig:conjugatearray} for our running example's conjugate array.
	\begin{figure}[tbp]
		\centering
		\begin{tabular}{|r||l|l|c||c|c|l|} 
			\hline
			$i$ & $\conj{\TT}{i}$ & $\mpdenc{\conj{\TT}{i}^\omega[..12]}$ & $\icarr{\TT}[i]$ & $\carr{\TT}[i]$ & $\mpdenc{\conj{\TT}{\carr{\TT}[i]}^\omega[..12]}$ & $\conj{\TT}{\carr{\TT}[i]}$\\
			\hline
			1 & $\mathtt{512}$ & $\mathtt{\Infty\Infty1131131131}$ & $\mathtt{9}$ & $\mathtt{8}$ & $\mathtt{\Infty11131113111}$ & $\mathtt{4478}$\\
			2 & $\mathtt{125}$ & $\mathtt{\Infty11311311311}$ & $\mathtt{3}$ & $\mathtt{9}$ & $\mathtt{\Infty11311131113}$ & $\mathtt{4784}$\\
			3 & $\mathtt{251}$ & $\mathtt{\Infty1\Infty113113113}$ & $\mathtt{7}$ & $\mathtt{2}$ & $\mathtt{\Infty11311311311}$ & $\mathtt{125}$ \\
			4 & $\mathtt{5363}$ & $\mathtt{\Infty\Infty1212121212}$ & $\mathtt{10}$ & $\mathtt{5}$ & $\mathtt{\Infty12121212121}$ & $\mathtt{3635}$ \\
			5 & $\mathtt{3635}$ & $\mathtt{\Infty12121212121}$ & $\mathtt{4}$ & $\mathtt{7}$ & $\mathtt{\Infty12121212121}$ & $\mathtt{3536}$ \\
			6 & $\mathtt{6353}$ & $\mathtt{\Infty\Infty1212121212}$ & $\mathtt{11}$ & $\mathtt{10}$ & $\mathtt{\Infty1\Infty111311131}$ & $\mathtt{7844}$ \\
			7 & $\mathtt{3536}$ & $\mathtt{\Infty12121212121}$ & $\mathtt{5}$ & $\mathtt{3}$ & $\mathtt{\Infty1\Infty113113113}$ & $\mathtt{251}$ \\
			8 & $\mathtt{4478}$ & $\mathtt{\Infty11131113111}$ & $\mathtt{1}$ & $\mathtt{11}$ & $\mathtt{\Infty\Infty1113111311}$ & $\mathtt{8447}$ \\
			9 & $\mathtt{4784}$ & $\mathtt{\Infty11311131113}$ & $\mathtt{2}$ & $\mathtt{1}$ & $\mathtt{\Infty\Infty1131131131}$ & $\mathtt{512}$ \\
			10 & $\mathtt{7844}$ & $\mathtt{\Infty1\Infty111311131}$ & $\mathtt{6}$ & $\mathtt{4}$ & $\mathtt{\Infty\Infty1212121212}$ & $\mathtt{5363}$ \\
			11 & $\mathtt{8447}$ & $\mathtt{\Infty\Infty1113111311}$ & $\mathtt{8}$ & $\mathtt{6}$ & $\mathtt{\Infty\Infty1212121212}$ & $\mathtt{6353}$ \\
			\hline
		\end{tabular}
		\caption{The conjugate array $\carr{\TT}$ of our running example $\TT = \{\mathtt{512},\mathtt{5363},\mathtt{4478}\}$.}
		\label{fig:conjugatearray}
	\end{figure}
	We define the \emph{conjugate range} $\crange{\TT}{P}$ of a pattern $P \in \Sigma^*$ of length $m$ in $\TT$ as a maximal interval $[\ell..r] \subseteq [1..n]$ such that $P \ctmatch \conj{\TT}{\carr{\TT}[i]}^\omega[..m]$ for every $i \in [\ell..r]$.
	Leveraging Lemma~\ref{lem:encctmatch}, we find that the conjugate range is well-defined.
	\begin{cor}\label{cor:uniquecrange}
		Let $\emptyset \neq \TT = \{\alltexts{T}{d}\}\subset \extsigma^+$, $n = \left|T_1\cdots T_d\right|$, $P \in \Sigma^*$, and $m = \left|P\right|$. 
		Then $P \ctmatch \conj{\TT}{\carr{\TT}[i]}^\omega[..m]$ if and only if $i \in \crange{\TT}{P}$.
	\end{cor}
	Note that $\crange{\TT}{\varepsilon} = [1..n]$.
	We call the computation of $\crange{\TT}{P}$ for some pattern $P\in\Sigma^*$ the \emph{backward search} for $P$ in $\TT$, and the length of $\crange{\TT}{P}$ is the solution to \Count{} by Corollary~\ref{cor:uniquecrange}.
	For our running example, $\crange{\TT}{\mathtt{643}} = \emptyset$ and $\crange{\TT}{\mathtt{5634}} = [6..7]$.
	Below we define the necessary tools to allow for an efficient backward search. 
	
	\subsection{LF-mapping}\label{subsection:lfmapping}
	
	We want to define a map that allows us to cycle backwards through each of the texts to be indexed by mapping to the position of a conjugate in the conjugate array from the position of its left rotation.
	However, if we want to represent this mapping space-efficiently, we have to be careful since we used tie-breaks within texts when we sorted the conjugates by a preorder.
	Our idea is to relax the requirements and define a map that allows us to cycle backwards through a text that is $\omega$-equal to the original to dodge any issues arising from our tie-breaks.
	For that, we want to cycle backwards in text order inside the roots of each $\rpdenc{..}$-encoded text. 
	Whenever we want to move backwards at the starting position of a root, we jump to its end position.
We express this backwards movement with the following permutation.
	For every $i \in [1..n]$ with $j = \min\{h \in [1..d] \mid \sum_{k = 1}^{h}n_k \geq i\}$, let
	\[
	\prev{\TT}{i} =
	\begin{cases*}
		i-1 + \wurz{\rpdenc{T_j}} & \text{if $\conj{\TT}{i} \cteq T_j$,}\\
		i- 1& \text{otherwise.}
	\end{cases*}
	\]
	For our running example, $\textup{prev}_\TT = (1\;3\;2)(4\;5)(6\;7)(8\;11\;10\;9)$.
	Here, we observe that we have two cycles corresponding to $T_2[1..2] = \mathtt{53}$ and $T_2[3..4] = \mathtt{63}$, 
	which are $\omega$-equal to the original text $T_2 = \mathtt{5363}$.
	Now we are ready to express the LF-mapping in terms of the function~$\textup{prev}_\TT$.
	The \emph{LF-mapping} $\lfarr{\TT}$ of $\TT$ is a string of length $n$ over $[1..n]$ such that $\lfmap{\TT}{i} = \icarr{\TT}[\prev{\TT}{\carr{\TT}[i]}]$ for each $i \in [1..n]$.
	Since the LF-mapping is a permutation, it has an inverse $\lfarr{\TT}^{-1}$, which we call the \emph{FL-mapping} of $\TT$ and denote by $\flarr{\TT}$.
	See Figure~\ref{fig:cbwt} for an example.
	The LF-mapping is at the core of the backward search.
	However, storing LF-mapping and FL-mapping in their plain form creates the need for two integer arrays of length $n$ and entries of $\lg n$ bits, which motivates the following encoding scheme.
	Let the \emph{rotational Cartesian tree signature encoding} $\rtsenc{V}$ of $V \in \extsigma^+$ denote a string of length $\left|V\right|$ over $\extsigma$ such that
	\[
	\rtsenc{V}[i] =
	\begin{cases*}
		V[i] & \text{if $V[i] = \$$,}\\
		\rank{\Infty}{\mpdenc{\rta{V}{i}}}{\left|V\right|} - \rank{\Infty}{\mpdenc{V[i]\cdot\rta{V}{i}}[2..]}{\left|V\right|} & \text{otherwise,}\\
	\end{cases*}
	\]
	for each $i \in [1..\left|V\right|]$, i.e., if $V[i] \neq \$$, then $\rtsenc{V}[i]$ reports the number of positions $j \in [1..\left|V\right|]$ such that $\rta{V}{i}[j] \geq V[i]$ and $\mpdenc{\rta{V}{i}}[j] = \Infty$.
	See Figure~\ref{fig:rotencoding} for an example. 
	Note that $\rtsenc{V}[i..\select{\$}{V}{1}] = \rtsenc{V[i..\select{\$}{V}{1}]}$ for each $V \in \extsigma^+$ satisfying $\rank{\$}{V}{\left|V\right|} \geq 1$ and $i \in [1..\select{\$}{V}{1}]$.
	\begin{lem}\label{lem:sumrtsenc}
		Let $V\in\extsigma^+$. 
		Then $\sum_{i = 1}^{j} \rtsenc{V}[i] \leq \left|V\right|$, where $j = \select{\$}{V}{1} - 1$ if $\rank{\$}{V}{\left|V\right|} \geq 1$, and $j = \left|V\right|$ otherwise.
	\end{lem}
	\begin{figure}[t]
		\centering
		\begin{tabular}{|c||c|c|c|c|c|} 
			\hline
			$i$ & $T_3[i]$ & $\rta{T_3}{i}$& $\mpdenc{\rta{T_3}{i}}$ & $\mpdenc{T_3[i]\cdot\rta{T_3}{i}}$ & $\rtsenc{T_3}[i]$  \\
			\hline
			1&$\mathtt{4}$&$\mathtt{4784}$&$\mathtt{\Infty113}$&$\mathtt{\Infty1113}$&$\mathtt{1}$\\
			2&$\mathtt{4}$&$\mathtt{7844}$&$\mathtt{\Infty1\Infty1}$&$\mathtt{\Infty1131}$&$\mathtt{2}$\\
			3&$\mathtt{7}$&$\mathtt{8447}$&$\mathtt{\Infty\Infty11}$&$\mathtt{\Infty1\Infty11}$&$\mathtt{1}$\\
			4&$\mathtt{8}$&$\mathtt{4478}$&$\mathtt{\Infty111}$&$\mathtt{\Infty\Infty111}$&$\mathtt{0}$\\
			\hline
		\end{tabular}
		\caption{Rotational Cartesian tree signature encoding $\rtsenc{T_3}$ of $T_3 = \mathtt{4478}$.}
		\label{fig:rotencoding}
	\end{figure}
	Taking advantage of both encodings, we investigate how the $\omega$-preorder of two strings changes if we rotate them.
	For notational convenience, let $\pi(V) = \rtsenc{V}[1]$ for every $V\in\extsigma^+$, and $\lcpcount{U}{W} = \rank{\Infty}{\mpdenc{U}}{\lcplength{\mpdenc{U}}{\mpdenc{W}}}$ for each $U, W \in \extsigma^*$.
	
	\begin{restatable}[{\cite[Lemma~3]{kim21compact}}]{lem}{rotaomega}\label{lem:rotationomegaorder}
		Let $V, U \in \extsigma^+$ such that $\rta{V}{1} \ctprec \rta{U}{1}$. 
		Then $V \ctprec U$ if and only if (a) $\pi(V) = \$$ or (b) $\pi(V) \geq \min\{e, \pi(U)\}$ and $\pi(U) \neq \$$, where $e = \lcpcount{\rta{V}{1}}{\rta{U}{1}}$.
	\end{restatable}
	
	Our representation of the LF- and FL-mapping consists of two strings $\ltarr{\TT}$ and $\ftarr{\TT}$, which are defined as follows.
	First, $\ltarr{\TT}$ is the string of length $n$ over $\extsigma$ such that $\ltarr{\TT}[i] = \pi(\conj{\TT}{\carr{\TT}[\lfmap{\TT}{i}]})$ for each $i\in[1..n]$.
	\begin{cor}\label{cor:lforder}
		Let $\emptyset \neq \TT \subset \extsigma^+$, $n$ the accumulated length of all texts in $\TT$, $i,j \in[1..n]$, and $i < j$. 
		If $\ltarr{\TT}[i] = \ltarr{\TT}[j]$, then $\lfmap{\TT}{i} < \lfmap{\TT}{j}$.
	\end{cor}
	Second, $\ftarr{\TT}$ is the string of length $n$ over $\extsigma$ such that $\ftarr{\TT}[\lfmap{\TT}{i}] = \ltarr{\TT}[i]$ for each $i \in [1..n]$.
	By what follows, $\ltarr{\TT}$ and $\ftarr{\TT}$ suffice to compute both LF- and FL-mapping of $\TT$.
	\begin{cor}\label{cor:fastlfmapcompute}
		Let $\emptyset \neq\TT \subset \extsigma^+$, $n$ the accumulated length of all texts in $\TT$, and $i \in [1..n]$. 
		Then $\lfmap{\TT}{i} = \select{\ltarr{\TT}[i]}{\ftarr{\TT}}{\rank{\ltarr{\TT}[i]}{\ltarr{\TT}}{i}}$ and $\flmap{\TT}{i} = \select{\ftarr{\TT}[i]}{\ltarr{\TT}}{\rank{\ftarr{\TT}[i]}{\ftarr{\TT}}{i}}$.
	\end{cor}
	In Figure~\ref{fig:cbwt} we present $\ftarr{\TT}$ and $\ltarr{\TT}$ of our running example.
	
	\subsection{Backward Search}\label{subsection:backwardsearch}
	
	The LF-mapping can be leveraged for the backward search by the following result, which is due to Lemma~\ref{lem:encctmatch} and Corollary~\ref{cor:uniquecrange}.
	\begin{lem}[{\cite[Lemma~6]{kim21compact}}]\label{lem:rangelarray}Let $\emptyset \neq \TT \subset \extsigma^+$, $P\in\Sigma^+$, $\left|P\right|=m$, $i \in [1..m]$, $h = \pi(P[i..]\cdot\$)$ and $e = \rank{\Infty}{\mpdenc{P[i..]}}{m-i+1}$. 
		For $j \in \crange{\TT}{P[i+1..]}$, $\lfmap{\TT}{j} \in \crange{\TT}{P[i..]}$ if and only if (a) $e > 1$ and $\ltarr{\TT}[j] = h$ or (b) $e = 1$ and $\ltarr{\TT}[j] \geq h$.
	\end{lem}
	At this point it is straight-forward to apply the techniques developed by Kim and Cho~\cite[Sections~5 and~6]{kim21compact} to define an index of $\TT \subset \Sigma^+$ that occupies $3n + o(n)$ bits of space and that solves \Count{} in $O(m)$ time, where $m$ is the pattern length.
	However, for brevity and in view of a space efficient construction of our proposed index and its extension, we will represent $\ftarr{\TT}$ and $\ltarr{\TT}$ by the dynamic data structure of Lemma~\ref{lemma:dynamicstringrmq}, and introduce an auxiliary string for the backward search. 
	Let $\lcpinfty{\TT}$ denote a string of length $n$ over $\Sigma$ such that $\lcpinfty{T}[1] = 0$ and $\lcpinfty{\TT}[i] = \lcpcount{\conj{\TT}{\carr{T}[i]}}{\conj{\TT}{\carr{T}[i-1]}}$ for each $i \in [2..n]$.
	
	\begin{restatable}{lem}{lcpcompute}\label{lemma:lcpinftycomputation}
		Let $\emptyset \neq \TT = \{\alltexts{T}{d}\}\subset\extsigma^+$, $n = \left|T_1\cdots T_d\right|$, $i,j \in [1..n]$, and $i < j$. 
		Then $\lcpcount{\conj{\TT}{\carr{\TT}[i]}}{\conj{\TT}{\carr{\TT}[j]}} = \min\{\lcpinfty{\TT}[k] \mid k \in[i+1..j]\} = \rnv{\lcpinfty{\TT}}{i+1}{j}{-1}$.
	\end{restatable}
	
	\begin{lem}\label{lemma:crange}
		Let $\emptyset \neq \TT \subset \extsigma^+$, $P \in \Sigma^+$, $m = \left|P\right|$, $i \in [1..m]$, $[\ell..r] = \crange{\TT}{P[i+1..]}$, $h =\pi(P[i..]\cdot\$)$, and $e = \rank{\Infty}{\mpdenc{P[i..]}}{m - i + 1}$. 
		Then Algorithm~\ref{alg:crangeupd} correctly computes $[\ell'..r'] = \crange{\TT}{P[i..]}$.
	\end{lem}
		\begin{algorithm}[tbp]
		\caption{
			Computing the conjugate range $\crange{\TT}{P[i..]}$. Here, $\emptyset \neq \TT \subset \extsigma^+$, $P \in \Sigma^+$, $m = \left|P\right|$, $i \in [1..m]$, $[\ell..r] = \crange{\TT}{P[i+1..]}$, $h =\pi(P[i..]\cdot\$)$, and $e =\rank{\infty}{\mpdenc{P[i..]}}{m - i + 1}$. 
		}
		\label{alg:crangeupd}
		\SetKwFunction{FMain}{crangeupd}
		\SetKwProg{Fn}{Function}{:}{}
		\Fn{\FMain{$e, h,[\ell..r], \ltarr{\TT}, \ftarr{\TT},\lcpinfty{\TT}$}}{
			$r' \gets r$\;{\label{alg:crangeupd:err}}
			\If(\tcp*[f]{Case~1 in Lemma~\ref{lemma:crange}}){$e > 1${\label{alg:crangeupd:firstcasestart}}}{{\label{alg:crangeupd:1start}}
				\If{$c \gets \rangecount{\ltarr{\TT}}{\ell}{r}{h}{h} \geq 1$}{{\label{alg:crangeupd:1cee}}
					$r' \gets \lfmap{\TT}{\select{h}{\ltarr{\TT}}{\rank{h}{\ltarr{\TT}}{r}}}$\;{\label{alg:crangeupd:1end}}
				}
			}
			\Else(\tcp*[f]{Case~2 in Lemma~\ref{lemma:crange}}){{\label{alg:crangeupd:2start}}
				\If{$c \gets \rangecount{\ltarr{\TT}}{\ell}{r}{h}{\sigma} \geq 1$}{{\label{alg:crangeupd:2cee}}
					$v \gets \rnv{\ltarr{\TT}}{\ell}{r}{h - 1}$\;{\label{alg:crangeupd:2val}}
					$x \gets \select{v}{\ltarr{\TT}}{\rank{v}{\ltarr{\TT}}{r}}$\;{\label{alg:crangeupd:2x}}
					$[\ell''..r''] \gets \maxi{\lcpinfty{\TT}}{x}{v + 1}$\;{\label{alg:crangeupd:2interval}}
					$y \gets \rangecount{\ltarr{\TT}}{\ell}{x-1}{h}{\sigma} + \rangecount{\ltarr{\TT}}{x+1}{r''}{h}{\sigma}$\;{\label{alg:crangeupd:2smaller}}
					$r' \gets \lfmap{\TT}{x} + c - (y + 1)$\;{\label{alg:crangeupd:2end}}
				}
			}
			\KwRet $[r' - c + 1..r']$\tcp*{returns $[r+1..r] = \emptyset$ if $c = 0$}{\label{alg:crangeupd:ell}}
		}
	\end{algorithm}\begin{proof}
		Let $c = \left|\crange{\TT}{P[i..]}\right|$.
		We show how Algorithm~\ref{alg:crangeupd} computes $c$ and $r'$ to obtain $[\ell'..r']$.
		
		\begin{cs1}
			Assume $e > 1$.
			This case is handled from Line~\ref{alg:crangeupd:1start} through Line~\ref{alg:crangeupd:1end}.
			By Lemma~\ref{lem:rangelarray}, $\lfmap{\TT}{j} \in \crange{\TT}{P[i..]}$ if and only if $\ltarr{\TT}[j] = h$ for each $j \in [\ell..r]$.
			We compute $c$ in Line~\ref{alg:crangeupd:1cee} and return an empty interval in Line~\ref{alg:crangeupd:ell} due to Line~\ref{alg:crangeupd:err} if $c \leq 0$.
			Assume $c \geq 1$.
			We compute the largest $x \in [\ell..r]$ such that $\ltarr{\TT}[x] = h$, and then $r' = \lfmap{\TT}{x}$ in Line~\ref{alg:crangeupd:1end}, with correctness following by Corollary~\ref{cor:lforder}.
		\end{cs1}
		
		\begin{cs2}
			Assume $e = 1$.
			This case is handled from Line~\ref{alg:crangeupd:2start} through Line~\ref{alg:crangeupd:2end}.
			By Lemma~\ref{lem:rangelarray}, $\lfmap{\TT}{j} \in \crange{\TT}{P[i..]}$ if and only if $\ltarr{\TT}[j] \geq h$ for each $j \in [\ell..r]$.
			Thus, we correctly compute $c$ in Line~\ref{alg:crangeupd:2cee}, and return an empty interval in Line~\ref{alg:crangeupd:ell} due to Line~\ref{alg:crangeupd:err} if $c \leq 0$.
			Assume $c \geq 1$.
			We compute the lowest value $v$ in $\ltarr{\TT}[\ell..r]$ greater than $h-1$ in Line~\ref{alg:crangeupd:2val} and determine the largest $x \in [\ell..r]$ such that $\ltarr{\TT}[x] = v$ in Line~\ref{alg:crangeupd:2x}.
			Let $[\ell''..r'']$ denote the interval computed in Line~\ref{alg:crangeupd:2interval}.
			Then $\lcpcount{\conj{\TT}{\carr{\TT}[j]}}{\conj{\TT}{\carr{\TT}[k]}} \geq v+1$ for each $j,k \in [\ell''..r'']$ by Lemma~\ref{lemma:lcpinftycomputation}.
			We compute $y = \lfmap{\TT}{x} - \ell'$.
			The following statements hold due to Lemma~\ref{lem:rotationomegaorder}.
			\begin{itemize}
				\item If $j \in [\ell..x-1]$ satisfies $\ltarr{\TT}[j] \geq h$, then $\lfmap{\TT}{j} < \lfmap{\TT}{x}$ since $v \leq \ltarr{\TT}[j]$. 
				\item If $j \in [x+1..r'']$  satisfies $\ltarr{\TT}[j] \geq h$, then $\lfmap{\TT}{j} < \lfmap{\TT}{x}$ since $v < \ltarr{\TT}[j]$ and $v < v + 1 \leq \lcpcount{\conj{\TT}{\carr{\TT}[j]}}{\conj{\TT}{\carr{\TT}[x]}}$.
				\item If $j \in [r''+1..r]$  satisfies $\ltarr{\TT}[j] \geq h$, then $\lfmap{\TT}{j} > \lfmap{\TT}{x}$ since we have $v \geq \lcpcount{\conj{\TT}{\carr{\TT}[j]}}{\conj{\TT}{\carr{\TT}[x]}}$.
			\end{itemize}
			We apply these results to compute $y$ in Line~\ref{alg:crangeupd:2smaller}, and then infer $r' = \lfmap{\TT}{x} + c - (y + 1)$ in Line~\ref{alg:crangeupd:2end}.
		\end{cs2}
	\end{proof}
	The next result is obtained from the computation of Demaine et al.'s~\cite{demaine14cartesian} Cartesian tree signature encoding of the given string.
	
	\begin{lem}\label{lem:preprocess}
		Given $P\in\extsigma^*$ of length $m$, we can process $P$ in $O(m)$ time such that we can subsequently compute $\pi(P[i..]\cdot\$)$ and $\rank{\Infty}{P[i..]}{m-i+1}$ in $O(1)$ time, for every $i \in [1..m]$.
	\end{lem}
	
	\begin{thm}\label{thm:indexexists}
		Let $\emptyset \neq \TT = \{\alltexts{T}{d}\} \subset \extsigma^+$ and $n = \left|T_1\cdots T_d\right|$.
		There exists a data structure occupying $O(n \lg \sigma)$ bits of space that solves \Count{} in $O(m \frac{\lg\sigma\lg n}{\lg\lg n})$ time, where $m$ is pattern length.
	\end{thm}
	
	\begin{proof}
		We represent $\ftarr{\TT}$, $\ltarr{\TT}$ and $\lcpinfty{\TT}$ by the data structure of Lemma~\ref{lemma:dynamicstringrmq}, which leads to the claimed space complexity.
		We preprocess a pattern $P\in\Sigma^+$ of length $m$ with Lemma~\ref{lem:preprocess}, and compute $\crange{\TT}{P[i..]}$ from $\crange{\TT}{P[i+1..]}$ for each $i \in [1..m]$ in descending order leveraging Lemma~\ref{lemma:crange}. 
		Since each conjugate range update takes $O(1)$ queries, the claimed complexity for solving \Count{} follows from Lemma~\ref{lemma:dynamicstringrmq}.
	\end{proof}
	We call the data structure of Theorem~\ref{thm:indexexists} the \emph{cBWT index} of $\TT$.
	The cBWT index of the running example is presented in Figure~\ref{fig:cbwt}.
	\begin{figure}[t]
		\centering
		\begin{tabular}{|c||c|l|c|c|c|c|l|} 
			\hline
			$i$ & $\carr{\TT}[i]$ & $\conj{\TT}{\carr{\TT}[i]}$ & $\lfarr{\TT}[i]$ & $\ftarr{\TT}[i]$ & $\ltarr{\TT}[i]$ & $\lcpinfty{\TT}[i]$ & $\mpdenc{\conj{\TT}{\carr{\TT}[i]}}$\\
			\hline
			 1 & $\mathtt{8}$ & $\mathtt{4478}$ & $\mathtt{8}$ & $\mathtt{1}$ & $\mathtt{0}$ & $\mathtt{0}$ & $\mathtt{\Infty111}$\\
			 2 & $\mathtt{9}$ & $\mathtt{4784}$ & $\mathtt{1}$ & $\mathtt{2}$ & $\mathtt{1}$ & $\mathtt{1}$ & $\mathtt{\Infty113}$ \\
			 3 & $\mathtt{2}$ & $\mathtt{125}$ & $\mathtt{9}$ & $\mathtt{2}$ & $\mathtt{0}$ & $\mathtt{1}$ & $\mathtt{\Infty11}$ \\
			 4 & $\mathtt{5}$ & $\mathtt{3635}$ & $\mathtt{10}$ & $\mathtt{2}$ & $\mathtt{0}$ & $\mathtt{1}$ & $\mathtt{\Infty121}$ \\
			 5 & $\mathtt{7}$ & $\mathtt{3536}$ & $\mathtt{11}$ & $\mathtt{2}$ & $\mathtt{0}$ & $\mathtt{1}$ & $\mathtt{\Infty121}$ \\
			 6 & $\mathtt{10}$ & $\mathtt{7844}$ & $\mathtt{2}$ & $\mathtt{1}$ & $\mathtt{2}$ & $\mathtt{1}$ & $\mathtt{\Infty1\Infty1}$ \\
			 7 & $\mathtt{3}$ & $\mathtt{251}$ & $\mathtt{3}$ & $\mathtt{1}$ & $\mathtt{2}$ & $\mathtt{1}$ & $\mathtt{\Infty1\Infty}$ \\
			 8 & $\mathtt{11}$ & $\mathtt{8447}$ & $\mathtt{6}$ & $\mathtt{0}$ & $\mathtt{1}$ & $\mathtt{1}$ & $\mathtt{\Infty\Infty11}$ \\
			 9 & $\mathtt{1}$ & $\mathtt{512}$ & $\mathtt{7}$ & $\mathtt{0}$ & $\mathtt{1}$ & $\mathtt{2}$ & $\mathtt{\Infty\Infty1}$ \\
			 10 & $\mathtt{4}$ & $\mathtt{5363}$ & $\mathtt{4}$ & $\mathtt{0}$ & $\mathtt{2}$ & $\mathtt{2}$ & $\mathtt{\Infty\Infty12}$ \\
			 11 & $\mathtt{6}$ & $\mathtt{6353}$ & $\mathtt{5}$ & $\mathtt{0}$ & $\mathtt{2}$ & $\mathtt{2}$ & $\mathtt{\Infty\Infty12}$ \\
			\hline
		\end{tabular}
		\caption{The cBWT index of our running example $\TT = \{\mathtt{512},\mathtt{5363},\mathtt{4478}\}$.
			For a complete overview, we refer to \cref{fig:cbwtfull} in \cref{app:fullexample}, which additionally depicts the FL-mapping, the $\rtsenc{..}$-encoding of the (sorted) conjugates, and the conjugates in text position order like in \cref{fig:conjugatearray}.
}
		\label{fig:cbwt}
	\end{figure}
	\begin{rem}\label{rem:kimconceptual}
		The (conceptual) integer index by Kim and Cho~\cite[Section~4.1]{kim21compact} emerges as a special case of the cBWT index if $d = 1$, $\rank{\$}{T_1}{n_1} = 1$ and $T_1[n_1] = \$$ by substituting the occurrence of $\$$ in $\ftarr{\TT}$ and $\ltarr{\TT}$ for $-1$. 
	\end{rem}
	
	\section{Construction in Compact Space}
	
	We start with the construction of the cBWT index of a single text satisfying convenient properties, and then show how to leverage that result to construct the cBWT index of arbitrary single texts.
	Subsequently, we show how an existing cBWT index can be extended by another text.
	Finally, we can leverage these two results step by step to construct the cBWT index of an arbitrary number of texts iteratively.
	
	\subsection{Single Text cBWT Index}\label{subsection:single}
	
	We need another technical result before we can tackle the construction of the cBWT index of a single text.
	
	\begin{restatable}{lem}{lcpchangerota}\label{lem:lcpinftyrota}
		Let $V,U \in \extsigma^+$, $\rta{V}{1} \ctprec \rta{U}{1}$, and $e = \lcpcount{\rta{V}{1}}{\rta{U}{1}}$.
		Then
		\[
		\lcpcount{V}{U} =
		\begin{cases*}
			0 & \text{if $\pi(U) = \$$ or $\pi(U) = \$$,}\\
			e - \pi(V) + 1 & \text{if $V \ctprec U$ and $\$ \neq \pi(V) = \pi(U) < e$,}\\
			1 & \text{otherwise.}
		\end{cases*}
		\]
	\end{restatable}
	
	\begin{lem}\label{lem:singlecBWT}
		Let $R \in \extsigma^+$, $\rho = \left|R\right|$, $R[\rho] = \$$, and $\rank{\$}{R}{\rho} = 1$.
		Algorithm~\ref{alg:singlecBWT} correctly computes $\icarr{\{bR\}}[\rho+1]$ and the cBWT index of $\{bR\}$ for each $b \in \Sigma$.
	\end{lem}
	\begin{algorithm}[tbp]
		\caption{
			Computing $c + 1 =\icarr{\{bR\}}[\rho+1]$ and updating the cBWT of $\{R\}$ to that of $\{bR\}$ for $b \in \Sigma$.
			Here, $R \in \extsigma^+$, $\rho = \left|R\right|$, $R[\rho] = \$$, $\rank{\$}{R}{\rho} = 1$, and $y = \icarr{\{R\}}[\rho]$.
		}
		\label{alg:singlecBWT}
		\SetKwFunction{FMain}{extend}
		\SetKwProg{Fn}{Function}{:}{}
		\Fn{\FMain{$\pi(bR),y,\ftarr{\{R\}},\ltarr{\{R\}},\lcpinfty{\{R\}}$}}{
$[\ell..r] \gets \maxi{\lcpinfty{\{R\}}}{y}{\pi(bR) + 1}$\;{\label{alg:singlecBWT:startleft}}
$c \gets \rangecount{\ltarr{\{R\}}}{\ell}{y-1}{\pi(bR)}{\sigma} + 2 + \rangecount{\ltarr{\{R\}}}{y+1}{r}{\pi(bR)+1}{\sigma}$\;\For{$i \gets \pi(bR)$ \KwTo $1$}{
				$r' \gets \ell - 1$\;{\label{alg:singlecBWT:rightborder}}
				$[\ell..r] \gets \maxi{\lcpinfty{\{R\}}}{y}{i}$\;{\label{alg:singlecBWT:leftborder}}
				$c \gets c + \rangecount{\ltarr{\{R\}}}{\ell}{r'}{i}{\sigma}$\;}
			$p, s \gets 1$\tcp*{$p$ stores $\lcpinfty{\{bR\}}[c+1]$, and, if $c<\rho$, $s$ stores $\lcpinfty{\{bR\}}[c+2]$}{\label{alg:singlecBWT:lcpstart}}
			\If{$\flmap{\{R\}}{c} < y$ and $\pi(bR) = \ftarr{\{R\}}[c]$}{
				$p \gets \rnv{\lcpinfty{\{R\}}}{\flmap{\{R\}}{c}+1}{y}{-1} - \pi(bR) + 1$\;
			}
			\If{$c < \rho$ and $\flmap{\{R\}}{c+1} > y$ and $\pi(bR) = \ftarr{\{R\}}[c]$}{
				$s \gets \rnv{\lcpinfty{\{R\}}}{y+1}{\flmap{\{R\}}{c+1}}{-1} - \pi(bR) + 1$\;
			}
			$\inst{\lcpinfty{\{R\}}}{c+1}{p}$\;
			\lIf{$c < \rho$}{{\label{alg:singlecBWT:lcpend}}
				$\lcpinfty{\{R\}}[c+2] \gets s$
			}
			$\inst{\ftarr{\{R\}}}{c+1}{\pi(bR)}$\;{\label{alg:updateks:flstart}}
			$\ltarr{\{R\}}[y] \gets \pi(bR)$\;
			$\inst{\ltarr{\{R\}}}{c+1}{\$}$\;{\label{alg:updateks:flend}}
			\KwRet $c + 1, \ftarr{\{R\}},\ltarr{\{R\}},\lcpinfty{\{R\}}$\tcp*{$\icarr{\{bR\}}[\rho+1]$ and the cBWT index of $\{bR\}$}
		}
	\end{algorithm}
	\begin{proof}
		For each $i \in [0..\pi(bR)]$, let $J_i = [\ell_i..r_i]$ maximal such that $r_{i} \leq \icarr{\{R\}}[\rho]$ and $\lcpcount{R}{\conj{\{R\}}{\carr{\{R\}}[j]}} = i$ for each $j \in J_i$, and $J_{\pi(bR)+1} = [\ell_{\pi(bR)+1}..r_{\pi(bR)+1}]$ maximal such that $\lcpcount{R}{\conj{\{R\}}{\carr{\{R\}}[j]}} \geq \pi(bR)+1$ for each $j \in J_{\pi(bR)+1} $, i.e., the left and right boundary of the former are computed in Line~\ref{alg:singlecBWT:leftborder} and Line~\ref{alg:singlecBWT:rightborder}, respectively, and the latter in Line~\ref{alg:singlecBWT:startleft}.
		The following statements are due to the location of the single $\$$ in each conjugate of $R$ and Lemma~\ref{lem:rotationomegaorder}.
		\begin{itemize}
			\item If $j \in [r_{\pi(bR)+1} + 1..\rho]$, then $bR \ctprec \conj{\{R\}}{\carr{\{R\}}[\lfmap{\{R\}}{j}]}$.
			\item If $j \in [\icarr{\{R\}}[\rho]+1..r_{\pi(bR)+1}]$, then $\conj{\{R\}}{\carr{\{R\}}[\lfmap{\{R\}}{j}]} \ctprec bR $ if and only if $\ltarr{\{R\}}[j] \geq \pi(bR)+1$.
			\item If $j = \icarr{\{R\}}[\rho]$, then $\conj{\{R\}}{\carr{\{R\}}[\lfmap{\{R\}}{j}]} \ctprec bR $.
			\item If $j \in [\ell_{\pi(bR)+1}..\icarr{\{R\}}[\rho]-1]$, then $\conj{\{R\}}{\carr{\{R\}}[\lfmap{\{R\}}{j}]} \ctprec bR $ if and only if $\ltarr{\{R\}}[j] \geq \pi(bR)$.
			\item If $i \in [1..\pi(bR)]$ and $j \in J_i$, then $\conj{\{R\}}{\carr{\{R\}}[\lfmap{\{R\}}{j}]} \ctprec bR $ if and only if $\ltarr{\{R\}}[j] \geq i$.
			\item $J_0 = [1..1]$ and  $(\$\cdot R)[..\rho] = \conj{\{R\}}{\carr{\{R\}}[\lfmap{\{R\}}{1}]} \ctprec bR $.
		\end{itemize}
		We apply these results from Line~\ref{alg:singlecBWT:startleft} through Line~\ref{alg:singlecBWT:leftborder} to compute the number $c$ of conjugates of $R$ smaller than $bR$ according to $\omega$-preorder.
		Since $\conj{\{R\}}{k}[..\select{\$}{\conj{\{R\}}{k}}{1}] = \conj{\{bR\}}{k+1}[..\select{\$}{\conj{\{R\}}{k}}{1}]$ for each $k \in [1..\rho]$ by assumption on $R$ and $b$,
		\[
		\conj{\{bR\}}{\carr{\{bR\}}[i]}[..x] =
		\begin{cases}
			\conj{\{R\}}{\carr{\{R\}}[i]}[..x] & \text{if  $i \in [1..c-1]$,}\\
			\conj{\{R\}}{\carr{\{R\}}[i-1]}[..x] & \text{if  $i-1 \in [c..\rho]$,}\\
			bR & \text{otherwise,}
		\end{cases}
		\]
		for each $i \in [1..\rho+1]$ with $x = \select{\$}{\conj{\{bR\}}{\carr{\{bR\}}[i]}}{1}$, i.e., $c$ is also the number of conjugates of $bR$ smaller than $bR$ according to $\omega$-preorder.
		Subsequently, Algorithm~\ref{alg:singlecBWT} updates the cBWT index of $\{R\}$ from Line~\ref{alg:singlecBWT:lcpstart} through Line~\ref{alg:updateks:flend}.
		By assumption on $b$ and $R$, $\lcpinfty{\{R\}}[..c] = \lcpinfty{\{bR\}}[..c]$ and $\lcpinfty{\{R\}}[c+2..] = \lcpinfty{\{bR\}}[c+3..]$.
		From Line~\ref{alg:singlecBWT:lcpstart} through Line~\ref{alg:singlecBWT:lcpend} we leverage Lemma~\ref{lemma:lcpinftycomputation} and Lemma~\ref{lem:lcpinftyrota} to compute $\lcpinfty{\{bR\}}[c+1]$ and, if $c < \rho$, $\lcpinfty{\{bR\}}[c+2]$, and update $\lcpinfty{\{R\}}$.
		By assumption on $R$, $\pi(\conj{\{R\}}{j}) = \pi(\conj{\{bR\}}{j+1})$ for each $j \in [1..\rho]$.
		Consequently, $\ftarr{\{R\}}[..c] = \ftarr{\{bR\}}[..c]$ and $\ftarr{\{R\}}[c+1..] = \ftarr{\{bR\}}[c+2..]$.
		Moreover, if we set $\ltarr{\{R\}}[\icarr{\{R\}}[\rho]] = \pi(bR)$, then $\ltarr{\{R\}}[..c] = \ltarr{\{bR\}}[..c]$ and $\ltarr{\{R\}}[c+1..] = \ltarr{\{bR\}}[c+2..]$.
		It remains to compute $\ftarr{\{bR\}}[c+1]$ and $\ltarr{\{bR\}}[c+1]$, which are $\pi(bR)$ and $\$$, respectively.
		The update of both $\ftarr{\{R\}}$ and $\ltarr{\{R\}}$ is done from Line~\ref{alg:updateks:flstart} through Line~\ref{alg:updateks:flend}.
	\end{proof}

	\begin{cor}\label{cor:kimindexconstruction}
		Let $R \in \extsigma^+$, $\rho = \left|R\right|$, $R[\rho] = \$$, and $\rank{\$}{R}{\rho} = 1$.
		The cBWT index of $\{R\}$ can be constructed in $O(\rho \lg \sigma )$ bits of space and $O(\rho\frac{\lg \sigma \lg \rho}{\lg\lg \rho})$ time.
	\end{cor}
	\begin{proof}
		We represent $\ftarr{\{\$\}} = \ltarr{\{\$\}} = \mathtt{\$}$ and $\lcpinfty{\{\$\}} = 0$ by the data structure of Lemma~\ref{lemma:dynamicstringrmq}.
		Then $\icarr{\{\$\}}[\left|\mathtt{\$}\right|] = 1$.
		We preprocess $R$ with Lemma~\ref{lem:preprocess} in $O(\rho)$ time to have access to $\pi(R[i..]) = \pi(\rta{R}{i-1})$ for each $i \in [1..\rho]$ in $O(1)$ time.
		For each $i \in [1..\rho-1]$ in descending order, we apply Algorithm~\ref{alg:singlecBWT} to compute $\icarr{\{R[i..]\}}[\left|R[i..]\right|]$ and extend the cBWT index of $\{R[i+1..]\}$ to that of $\{R[i..]\}$.
		Since the extension of the cBWT index of $\{R[i+1..]\}$ to that of $\{R[i..]\}$ takes $O(\pi(R[i..]))$ queries for each $i \in [1..\rho-1]$, the construction of the cBWT index of $\{R[1..]\} = \{R\}$ takes a total of $O(\rho)$ queries by Lemma~\ref{lem:sumrtsenc}.
		The claimed complexities then follow by Lemma~\ref{lemma:dynamicstringrmq}.
	\end{proof}
	
	\begin{rem}	
		Due to Remark~\ref{rem:kimconceptual}, the previous statement is the online construction of the (conceptual) integer index presented by Kim and Cho~\cite[Section~4.1]{kim21compact}.
		For that, we substitute~$\$$ with~$-1$ in both $\ftarr{\TT}$ and $\ltarr{\TT}$.
		From the cBWT index, we can construct their $3n + o(n)$ bits of space representation~\cite[Section~5]{kim21compact} directly~\cite[Section~A.4]{kim21compact}
		while retaining the complexities as stated in Corollary~\ref{cor:kimindexconstruction}.
	\end{rem}
	
	\begin{cor}\label{cor:singleindexconstr}
		Let $T \in \Sigma^+$ and $\left|T\right| = n$.
		Then the cBWT index of $\{T\}$ can be constructed in $O(n \lg \sigma )$ bits of space and $O(n\frac{\lg \sigma \lg n}{\lg\lg n})$ time.
	\end{cor}
	\begin{proof}
		Let $R = T^4\$$.
		We construct the cBWT index of $\{R\}$ within the claimed complexities by Corollary~\ref{cor:kimindexconstruction}.
		During construction, we build a bit string $Y$ of length $4n+1$ such that $Y[i] = 1$ if and only if $\carr{\{R\}}[i] \in [2..n+1]$.
		By choice of $R$,
		$\conj{\{T\}}{i}^\omega[..3n] = \conj{\{R\}}{i}[..3n]$ for each $i \in [1..n]$, and $\conj{\{T\}}{1}^\omega[..3n] = \conj{\{R\}}{n+1}[..3n]$.	
		Then by choice of $Y$ and Lemma~\ref{lem:equality3z},
		\begin{itemize}
			\item $\conj{\{T\}}{\carr{\{T\}}[i]}^\omega[..3n] = \conj{\{R\}}{\carr{\{R\}}[\select{1}{Y}{i}]}[..3n]$,
			\item $\pi(\conj{\{T\}}{\carr{\{T\}}[i]}) = \pi(\conj{\{R\}}{\carr{\{R\}}[\select{1}{Y}{i}]})$, and
			\item $\pi(\conj{\{T\}}{\carr{\{T\}}[\lfmap{\{T\}}{i}]}) = \pi(\conj{\{R\}}{\carr{\{R\}}[\lfmap{\{R\}}{\select{1}{Y}{i}}]})$,
		\end{itemize}
		for each $i \in [1..n]$.
		Thus, $\ftarr{\{T\}}[i] = \ftarr{\{R\}}[\select{1}{Y}{i}]$ and $\ltarr{\{T\}}[i] = \ltarr{\{R\}}[\select{1}{Y}{i}]$ for each $i \in [1..n]$.
		Moreover, $\lcpinfty{\{T\}}[i] = \rnv{\lcpinfty{\{R\}}}{\select{1}{Y}{i-1}+1}{\select{1}{Y}{i}}{-1}$ for each $i \in [2..n]$ due to Lemma~\ref{lemma:lcpinftycomputation}, and $\lcpinfty{\{T\}}[1] = 0$.
		Hence, we can transform the cBWT index of $\{R\}$ to index $\{T\}$ with $O(n)$ queries.
		The claimed complexities then follow by Lemma~\ref{lemma:dynamicstringrmq}.
	\end{proof}

	\subsection{Extending an Existing cBWT Index}
	
	Let $\emptyset \neq \RT \subset \Sigma^+$ be a non-empty set of strings and $\rho = \left|\carr{\RT}\right|$ denote the accumulated length of all strings in $\RT$.
	In this subsection, we assume we have the cBWT index of $\RT$ at hand, 
	and want to extend it by another text $S \in \Sigma^+$ of length $\lambda$ to index $\ST = \RT \cup \{S\}$.
	To achieve this, we compute the integers $\cnt{\RT}{\rta{S}{i}}$, $\plcp{\RT}{\rta{S}{i}}$ and $\slcp{\RT}{\rta{S}{i}}$ for each $i \in [1..\lambda]$, whose definitions follow.
	Let $\cnt{\RT}{V} = \{i\in[1..\rho] \mid \conj{\RT}{\carr{\RT}[i]} \ctpeq V\}$,
	\[
	\begin{split}
		\plcp{\RT}{V} &= 
		\begin{cases*}
			-1 & \text{if $\cnt{\RT}{V} = 0$,} \\
			\lcpcount{\conj{\RT}{\carr{\RT}[\cnt{\RT}{V}]}}{V} & \text{otherwise, and}
		\end{cases*} \\
		\slcp{\RT}{V} &= 
		\begin{cases*}
			\lcpcount{\conj{\RT}{\carr{\RT}[\cnt{\RT}{V}+1]}}{V} & \text{if $\cnt{\RT}{V} < \rho$,}\\
			-1 & \text{otherwise,}
		\end{cases*} 
	\end{split}
	\]
	for each $V \in \extsigma^+$. 
	We apply the same technique as in Lemma~\ref{lem:singlecBWT} to obtain the following result, whose details are deferred to \cref{app:lem:next}.

	\begin{restatable}{lem}{nextstatement}
		\label{lem:next:statement}
		Let $\emptyset \neq \RT \subset \Sigma^+$, $\rho = \left|\carr{\RT}\right|$, and $V \in \extsigma^+$.
		With the cBWT index of $\RT$, we can compute $\cnt{\RT}{V}$, $\plcp{\RT}{V}$ and $\slcp{\RT}{V}$ from $\pi(V)$, $\cnt{\RT}{\rta{V}{1}}$, $\plcp{\RT}{\rta{V}{1}}$ and $\slcp{\RT}{\rta{V}{1}}$ in $O((1 + \pi(V)) \frac{\lg \sigma \lg \rho}{\lg\lg \rho})$ time if $\pi(V) \neq \$$.
	\end{restatable}

	For the iterative construction, we augment the cBWT index by a dynamic bit string $E$ with the invariant that $E$ is zeroed before and after each extension with a new string.
	We use $E$ to temporarily mark modified parts in the cBWT such that we retain the functionality of the index even during construction.
	The length of $E$ is the number $\rho$ of characters indexed by the cBWT index, which we call in the next lemma the \emph{augmented cBWT index} for clarity.
	
	\begin{lem}\label{lem:cBWTextend}
		Let $\emptyset \neq \RT \subset \Sigma^+$, $\rho = \left|\carr{\RT}\right|$, $S \in \Sigma^+$, $\ST = \RT \cup \{S\}$, $\lambda = \left|S\right|$, and $z = \max \{\left|V\right| \mid V \in \ST\}$.
		The augmented cBWT index of $\RT$ can be extended to index~$\ST$ in $O((\lambda + \rho) \lg \sigma)$ bits of space and $O(z \frac{\lg \sigma \lg (\lambda + \rho)}{\lg\lg (\lambda + \rho)})$ time.
	\end{lem}
	\begin{proof}
		We compute the cBWT index of $\{S\}$ within the claimed complexities by Corollary~\ref{cor:singleindexconstr}
		Then we compute $\cnt{\RT}{\rta{S}{i}}$, $\plcp{\RT}{\rta{S}{i}}$ and $\slcp{\RT}{\rta{S}{i}}$ for each $i \in [1..\lambda]$.
		
		\begin{cs1}
			Assume $\rta{S}{i} \cteq R$ for some $R \in \RT$ and $i \in [1..\lambda]$.
			Then $\cnt{\RT}{\rta{S}{i}} = r$, $\plcp{\RT}{\rta{S}{i}} = \rank{\Infty}{\mpdenc{\rta{S}{i}}}{\lambda}$, and, if $r < \rho$, $\slcp{\RT}{\rta{S}{i}} = \lcpinfty{\RT}[r+1]$ by Lemma~\ref{lem:equality3z}, where $i \in [1..\lambda]$ and $r$ is the right boundary of $\crange{\RT}{\rta{S}{i}^\omega[..3z]}$.
			Since $\rank{\Infty}{\mpdenc{\rta{S}{i-1}}}{\lambda} = \rank{\Infty}{\mpdenc{S^\omega[i..4z]}}{4z-i+1}$ for each $i \in [2..\lambda+1]$, the last $\lambda $ steps of the backward search for $S^\omega[2..4z]$ yield the necessary values to compute $\cnt{\RT}{\rta{S}{i}}$, $\plcp{\RT}{\rta{S}{i}}$ and $\slcp{\RT}{\rta{S}{i}}$ for all $i \in [1..\lambda]$ in descending order within the claimed complexities by Theorem~\ref{thm:indexexists}.
		\end{cs1}
		
		\begin{cs2}
			Assume $\rta{S}{i} \not\cteq R$ for each $R \in \RT$ and $i \in [1..\lambda]$.
			Let $V = S^\omega[..4z]\cdot\$$.
			By assumption and Corollary~\ref{cor:prec3z}, $\cnt{\RT}{\rta{S}{i}} = \cnt{\RT}{\rta{V}{i}}$, $\plcp{\RT}{\rta{S}{i}} = \plcp{\RT}{\rta{V}{i}}$, and $\slcp{\RT}{\rta{S}{i}} = \slcp{\RT}{\rta{V}{i}}$ for each $i \in [1..\lambda]$.
			We preprocess $V$ with Lemma~\ref{lem:preprocess} within the claimed complexities.
			Since we have $\cnt{\RT}{\rta{V}{4z}} = 0$, $\plcp{\RT}{\rta{V}{4z}} = -1$ and $\slcp{\RT}{\rta{V}{4z}} = 0$ by construction, we can now use Lemma~\ref{lem:next:statement}
			to compute $\cnt{\RT}{\rta{V}{j}}$, $\plcp{\RT}{\rta{V}{j}}$ and $\slcp{\RT}{\rta{V}{j}}$ for each $j \in [1..4z-1]$ in descending order.
			Hence, we obtain $\cnt{\RT}{\rta{S}{i}}$, $\plcp{\RT}{\rta{S}{i}}$ and $\slcp{\RT}{\rta{S}{i}}$ for all $i \in [1..\lambda]$ in descending order within the claimed complexities by Lemma~\ref{lem:sumrtsenc}.
			This concludes Case~2.
		\end{cs2}
	
		Leveraging the backward search on the cBWT index of $\{S\}$, we store the $\textup{plcp}_\RT^\infty$-values and $\textup{slcp}_\RT^\infty$-values in lex-order, which takes $O(\lambda \lg \sigma)$ bits of space.
		To store the $\textup{cnt}_\RT$-values in lex-order, compact space, and stay within the claimed time complexity, we utilize $E$, which we assume to be represented by the data structure of Lemma~\ref{lemma:dynamicstringrmq}.
		For each $i \in [1..\lambda]$ in descending order, we call $\inst{E}{\select{0}{E}{\cnt{\RT}{\rta{S}{i}}}+1}{1}$, and discard $\cnt{\RT}{\rta{S}{i}}$ once $\cnt{\RT}{\rta{S}{i-1}}$ has been computed.
		Then $\cnt{\RT}{\conj{\{S\}}{\carr{\{S\}}[i]}} = \rank{0}{E}{\select{1}{E}{i}}$.
		The following statements for each $i \in [1..\lambda + \rho]$ are due to construction.
		\begin{itemize}
			\item If $E[i] = 0$, then $\ftarr{\ST}[i] = \ftarr{\RT}[\rank{0}{E}{i}]$ and $\ltarr{\ST}[i] = \ltarr{\RT}[\rank{0}{E}{i}]$.
			\item If $E[i] = 1$, then $\ftarr{\ST}[i] = \ftarr{\{S\}}[\rank{1}{E}{i}]$ and $\ltarr{\ST}[i] = \ltarr{\{S\}}[\rank{1}{E}{i}]$.
			\item $\lcpinfty{\ST}[1] = 0$.
			\item If $i \geq 2$, $E[i-1] = 1$ and $E[i] = 1$, then $\lcpinfty{\ST}[i] = \lcpinfty{\{S\}}[\rank{1}{E}{i}]$.
			\item If $i \geq 2$, $E[i-1] = 1$ and $E[i] = 0$, then $\lcpinfty{\ST}[i] = \slcp{\RT}{\conj{\{S\}}{\carr{\{S\}}[\rank{1}{E}{i}]}}$.
			\item If $i \geq 2$, $E[i-1] = 0$ and $E[i] = 1$, then $\lcpinfty{\ST}[i] = \plcp{\RT}{\conj{\{S\}}{\carr{\{S\}}[\rank{1}{E}{i}]}}$.
			\item If $i \geq 2$, $E[i-1] = 0$ and $E[i] = 0$, then $\lcpinfty{\ST}[i] = \lcpinfty{\RT}[\rank{0}{E}{i}]$.
		\end{itemize}
		Consequently, $O(\lambda)$ queries and operations are needed to update the cBWT index of $\RT$ to index~$\ST$.
		Finally, we zero $E$ with $O(\lambda)$ queries.
		The claimed complexities now follow from Lemma~\ref{lemma:dynamicstringrmq}.
	\end{proof}

	We are finally able to state the main result of this section.

	\begin{thm}\label{thm:cBWTconstruction}
		Let $\emptyset \neq \TT = \{\alltexts{T}{d}\} \subset \Sigma^+$ and $n = \left|T_1\cdots T_d\right|$.
		If $\left|T_{k-1}\right| \leq \left|T_k\right|$ for each $k \in [2..d]$, then the cBWT index of $\TT$ can be constructed in $O(n\lg\sigma)$ bits of space and $O(n\frac{\lg\sigma\lg n}{\lg\lg n})$ time.
	\end{thm}
	\begin{proof}
		Let $\left|T_{k-1}\right| \leq \left|T_k\right|$ for each $k \in [2..d]$, and let $E_j$ denote a zeroed bit string of length $\left|T_1 \cdots T_j\right|$ for each $j \in [1..d]$.
		First, we construct $E_1$ and the cBWT index of $\{T_1\}$ within the claimed complexities by Lemma~\ref{cor:singleindexconstr}, where each string is represented by the data structure of Lemma~\ref{lemma:dynamicstringrmq}.
		Subsequently, we iteratively extend the cBWT index of $\{\alltexts{T}{k-1}\}$ augmented by $E_{k-1}$ to the cBWT index of $\{\alltexts{T}{k}\}$ augmented by $E_k$ for each $k \in [2..d]$ in ascending order leveraging Lemma~\ref{lem:cBWTextend}.
		Then the cBWT index of $\TT$ is constructed in $O((n + \left|T_d\right|) \lg \sigma) = O(n \lg\sigma)$ bits of space and $O((\sum_{k=1}^{d}\left|T_k\right|)\frac{\lg\sigma\lg n}{\lg\lg n}) = O(n \frac{\lg\sigma\lg n}{\lg\lg n})$ time.
	\end{proof}

\bibliographystyle{plain}

	\clearpage
\appendix

	\section{Proofs}\label{app:proofs}
	
	\rotaparent*
	\begin{proof}Let $\rpdenc{V} = \rpdenc{U}$ and $\left|V\right| = \left|U\right|$. 
		Assume $V^2 \not\approx_\textup{ct} U^2$. 
		By Lemma~\ref{lem:encctmatch}, there exists some $1 \leq i \leq \left|V\right|$ such that $\mpdenc{V}[i] \neq \mpdenc{U}[i]$. 
		If $\{\mpdenc{V}[i], \mpdenc{U}[i]\} \subset \extsigma$, then $\mpdenc{V}[i] = \mpdenc{V^2}[i + \left|V\right|] = \rpdenc{V}[i] = \rpdenc{U}[i] = \mpdenc{U^2}[i + \left|U\right|] = \mpdenc{U}[i]$. 
		Hence, either $\mpdenc{V}[i] = \Infty$ or $\mpdenc{U}[i] = \Infty$. 
		Without loss of generality, $\mpdenc{U}[i] = \Infty$. 
		But then $\rpdenc{U}[i] = \mpdenc{U^2}[i + \left|U\right|] > \mpdenc{V^2}[i + \left|V\right|] = \rpdenc{V}[i]$, a contradiction.
		Hence, $V^2 \ctmatch U^2$. 
		The remaining cases and opposite direction are immediate.
	\end{proof}

	For a string $X$, we call $1 \leq p\leq \left|X\right|$ satisfying $X[q] = X[q+p]$ for every $q \in [1..\left|X\right| - p]$ a \emph{period} of $X$.
	
	\begin{lem}[Weak Periodicity Lemma]\label{lem:weakperiodicity}
		Let $p$ and $q$ be two periods of a string $X$. If $p+q \leq \left|X\right|$, then $\gcd(p,q)$ is also a period of $X$.
	\end{lem}
	
	\equality*
	\begin{proof}
		Without loss of generality, $\left|U\right| = z$. 
		Let $\left|V\right|= y$, $\left|\wurz{\rpdenc{V}}\right| = p$ and $\left|\wurz{\rpdenc{U}}\right| = q$.
		Assume $\mpdenc{V^{\omega}[..3z]}  = \mpdenc{U^{\omega}[..3z]}$. 
		By Lemma~\ref{lem:rotaencodingmatch},
		\[
		\rpdenc{\rta{V}{z}}^\omega[..2z] = \rpdenc{\rta{U}{z}}^\omega[..2z] = \rpdenc{U^2}.
		\]
		Consequently, both $p$ and $z$ are periods of $\rpdenc{U^2}$. 
		Since $p + z \leq 2z$, we can apply Lemma~\ref{lem:weakperiodicity} and find that $\gcd(p,z)$  is a period of $\rpdenc{U^2}$. 
		As $\gcd(p,z)$ divides $z$, $\rpdenc{U}$ can be formed by repeating $\rpdenc{U}[..\gcd(p,z)]$ an integral number of times. This implies $\gcd(p,z) \geq q$, i.e. $y \geq p \geq q$. Also by assumption and Lemma~\ref{lem:rotaencodingmatch},
		\[
		\rpdenc{V^2} = \rpdenc{\rta{V}{z}}^\omega[..2y] = \rpdenc{\rta{U}{y}}^\omega[..2y].
		\]
		Since $q \leq y$, both $q$ and $y$ are periods of $\rpdenc{V^2}$. 
		Then $q + y \leq 2y$ and Lemma~\ref{lem:weakperiodicity} imply that $\gcd(q,y)$ is a period of $\rpdenc{V^2}$. 
		As $\gcd(q,y)$ divides $y$, $\rpdenc{V}[..\gcd(q,y)]$ can be repeated an integral number of times to form $\rpdenc{V}$. 
		This implies $\gcd(q,y) \geq p$ and consequently $q \geq p$. 
		Thus, $p = q$ and the statement follows. 
		The opposite direction is immediate.
	\end{proof}
	
	\rotaomega*
	\begin{proof}
		We show the statement for $\min\{\pi(V),\pi(U)\} \in \Sigma$ since it is non-trivial.
		Let  $z = \max\{\left|V\right|,\left|U\right|\}$ and $\lambda = \lcplength{\mpdenc{\rta{V}{1}^\omega[..3z]}}{\mpdenc{\rta{U}{1}^\omega[..3z]}}$. 
		Then $\mpdenc{\rta{V}{1}^{\omega}[..3z]}  < \mpdenc{\rta{U}{1}^{\omega}[..3z]}$ by Corollary~\ref{cor:prec3z}, $\lambda < 3z$,  and $\select{\Infty}{\mpdenc{\rta{V}{1}}}{e} \leq \min\{\left|V\right|,\left|U\right|\}$.
		
		\begin{rdir}
			Assume $\pi(V) < \min\{e,\pi(U)\}$. 
			Let $i = \select{\Infty}{\mpdenc{V}}{2}$.
			Then $\mpdenc{V}[..i-1] = \mpdenc{U}[..i-1]$ and $\mpdenc{V}[i] = \Infty \neq i-1 = \mpdenc{U}[i]$, i.e. $\mpdenc{U} < \mpdenc{V}$. 
			Consequently, $\mpdenc{U^{\omega}[..3z]}  < \mpdenc{V^{\omega}[..3z]}$, which implies $U \ctprec V$ by Corollary~\ref{cor:prec3z}.
			The statement follows by contraposition.
		\end{rdir}
		
		\begin{ldir}
			We consider three different cases.
		\end{ldir}
	
		\begin{cs1}
			Assume $\min\{e,\pi(V)\}> \pi(U)$. 
			Let $i = \select{\Infty}{\mpdenc{U}}{2}$. 
			Then $\mpdenc{V}[..i-1] = \mpdenc{U}[..i-1]$ and $\mpdenc{V}[i] = i-1 \neq \Infty = \mpdenc{U}[i]$, i.e. $\mpdenc{V} < \mpdenc{U}$. 
			Consequently, $\mpdenc{V^{\omega}[..3z]}  < \mpdenc{U^{\omega}[..3z]}$, which implies $V \ctprec U$ by Corollary~\ref{cor:prec3z}.
		\end{cs1}\\
		\begin{cs2}
			Assume $\min\{\pi(V), \pi(U)\} \geq e$. 
			Then we have $\mpdenc{V^\omega[..3z]}[..\lambda+1] = \mpdenc{U^\omega[..3z]}[..\lambda+1]$.
			Since $\rta{V}{1} \ctprec \rta{U}{1}$, $\lambda +2 \leq 3z$ by Lemma~\ref{lem:equality3z}.
			If $\mpdenc{\rta{U}{1}^\omega[..3z]}[\lambda+1] = \Infty$, $\mpdenc{V^\omega[..3z]}[\lambda+2] \leq \lambda < \lambda + 1=\mpdenc{U^\omega[..3z]}[\lambda+2]$.
			If $\mpdenc{\rta{U}{1}^\omega[..3z]}[\lambda] \neq \Infty$, then $\mpdenc{V^\omega[..3z]}[\lambda+2] = \mpdenc{\rta{V}{1}^\omega[..3z]}[\lambda+1] <  \mpdenc{\rta{U}{1}^\omega[..3z]}[\lambda+1] =\mpdenc{U^\omega[..3z]}[\lambda+2]$.
			Thus, $V \ctprec U$ by Corollary~\ref{cor:prec3z}.
		\end{cs2}\\
		\begin{cs3}
			Assume $e > \pi(V) = \pi(U)$. 
			Then $\mpdenc{V^\omega[..3z]}[..\lambda+1] = \mpdenc{U^\omega[..3z]}[..\lambda+1]$ and $\lambda \leq z$.
			Moreover, $\mpdenc{V^\omega[..3z]}[\lambda+2] = \mpdenc{\rta{V}{1}^\omega[..3z]}[\lambda+1] < \mpdenc{\rta{U}{1}^\omega[..3z]}[\lambda+1] = \mpdenc{U^\omega[..3z]}[\lambda+2]$.
			Hence, $V \ctprec U$ by Corollary~\ref{cor:prec3z}.
		\end{cs3}\\
		Since we exhausted all possible cases for $\pi(V) \geq \min\{e,\pi(U)\}$, the statement follows.
	\end{proof}
	
	\lcpcompute*
	\begin{proof}
		Let $U,V,W\in\extsigma^+$ and $z = \max\{\left|U\right|, \left|V\right|,\left|W\right|\}$.
		Then $\lcplength{\mpdenc{U^\omega[..3z]}}{\mpdenc{W^\omega[..3z]}} = \min\{\lcplength{\mpdenc{U^\omega[..3z]}}{\mpdenc{V^\omega[..3z]}}, \lcplength{\mpdenc{V^\omega[..3z]}}{\mpdenc{W^\omega[..3z]}}\}$ if we have $\mpdenc{U^\omega[..3z]} \leq \mpdenc{V^\omega[..3z]}$ and $\mpdenc{V^\omega[..3z]} \leq \mpdenc{W^\omega[..3z]}$.
		Then the statement follows by Lemma~\ref{lem:equality3z} and the definition of $\lcpinfty{\TT}$.
	\end{proof}
	
	\lcpchangerota*
	\begin{proof}
		If $\pi(W) = \$$ for any $W \in \extsigma^+$, then $\rank{\Infty}{\mpdenc{W}}{\left|W\right|} = 0$.
		It remains to show the cases where both $\pi(V)$ and $\pi(U)$ do not evaluate to $\$$.
		Thus, assume $\$ < \min\{\pi(V),\pi(U)\}$, i.e., $\lcpcount{U}{V} \geq 1$.
		By assumption, $V \not\cteq U$.
		
		\begin{cs1}
			Assume $U \ctprec V$.
			Then $\pi(V) < e$ and $\pi(V) < \pi(U)$ by Lemma~\ref{lem:rotationomegaorder}.
			Consequently, $\mpdenc{V}[\select{\Infty}{\mpdenc{V}}{2}] = \Infty > \select{\Infty}{\mpdenc{V}}{2} -1 =  \mpdenc{U}[\select{\Infty}{\mpdenc{V}}{2}]$, i.e., $\lcpcount{V}{U} = 1$.
		\end{cs1}
		
		\begin{cs2}
			Assume $V \ctprec U$.
			Then $\pi(V) \geq e$ or $\pi(V) \geq \pi(U)$ by Lemma~\ref{lem:rotationomegaorder}.
			
			\begin{cs21}
				Assume $e \geq \pi(U) = \pi(V)$.
				Then $\lcplength{\mpdenc{V}}{\mpdenc{U}} = \lcplength{\mpdenc{\rta{V}{1}}}{\mpdenc{\rta{U}{1}}} + 1$.
				Hence, $\lcpcount{V}{U} = e - \pi(V) + 1$, which evaluates to $1$ if and only if $e = \pi(V) = \pi(U)$.
			\end{cs21}
			
			\begin{cs22}
				Assume $\min\{e,\pi(V)\} > \pi(U)$.
				Then $\mpdenc{V}[\select{\Infty}{\mpdenc{U}}{2}] = \select{\Infty}{\mpdenc{U}}{2} - 1 < \Infty = \mpdenc{U}[\select{\Infty}{\mpdenc{U}}{2}]$.
				Consequently, $\lcpcount{V}{U} = 1$.
			\end{cs22}
			
			\begin{cs23}
				Assume $e \leq \min\{\pi(V),\pi(U)\}$.
				If we have $\lcplength{\mpdenc{\rta{V}{1}}}{\mpdenc{\rta{U}{1}}} = \min\{\left|U\right|,\left|V\right|\}$, then $\lcpcount{V}{U} = \rank{\Infty}{\mpdenc{V}}{\min\{\left|V\right|, \left|U\right|\}}$ = 1.
				Assume $\min\{\left|U\right|,\left|V\right|\} > \lcplength{\mpdenc{\rta{V}{1}}}{\mpdenc{\rta{U}{1}}}$.
				Then there exists some $x \in [1..\min\{\left|U\right|,\left|V\right|\}-1]$ such that $\mpdenc{\rta{V}{1}}[..x-1] = \mpdenc{\rta{U}{1}}[..x-1]$ and $\rta{V}{1}[x] < \rta{U}{1}[x]$.
				If $\rta{U}{1}[x] = \Infty$, then $\mpdenc{U}[x+1] \geq x > \mpdenc{V}[x+1]$.
				If $\rta{U}{1}[x] \neq \Infty$, then $\mpdenc{U}[x+1] = \rta{U}{1}[x] > \rta{V}{1}[x] = \mpdenc{V}[x+1]$.
				Either way, $\lcpcount{V}{U} = \rank{\Infty}{\mpdenc{V}}{x} = 1$.
			\end{cs23}
		\end{cs2}
	\end{proof}

	\clearpage
	\section{Computation of Helper Values}\label{app:lem:next}
	We here describe algorithmic details in how to obtain the claims of \cref{lem:next:statement}.

	\begin{algorithm}
		\caption{
			Computing $c = \cnt{\RT}{V}$, $p = \plcp{\RT}{V}$ and $s = \slcp{\RT}{V}$.
			Here, $\emptyset \neq \RT \subset \Sigma^+$, $\rho = \left|\carr{\RT}\right|$, $V \in \extsigma^+$, and $y = \cnt{\RT}{\rta{V}{1}}$.
}
		\label{alg:next}
		\SetKwFunction{FMain}{next}
		\SetKwProg{Fn}{Function}{:}{}
		\Fn{\FMain{$\pi(V),y,\plcp{\RT}{\rta{V}{1}},\slcp{\RT}{\rta{V}{1}},\ftarr{\RT},\ltarr{\RT},\lcpinfty{\RT}$}}{
		\lIf{$\pi(V) = \$$}{\label{alg:next:specialsymbol}
			\KwRet $0,-1,0$
		}
		\lIf{$\slcp{\RT}{\rta{V}{1}} \geq \pi(V)+1$}{\label{alg:next:cntstart}
			$[\ell..r] \gets \maxi{\lcpinfty{\RT}}{y+1}{\pi(V) + 1}$
		}
		\lElseIf{$\plcp{\RT}{\rta{V}{1}} \geq \pi(V)+1$}{
			$[\ell..r] \gets \maxi{\lcpinfty{\RT}}{y}{\pi(V) + 1}$
		}
		\lElse{
			$[\ell..r] \gets [y+1..y]$
		}
		$c \gets \rangecount{\ltarr{\RT}}{y+1}{r}{\pi(V)+1}{\sigma} + \rangecount{\ltarr{\RT}}{\ell}{y}{\pi(V)}{\sigma}$\;
		\For{$i \gets \min\{\pi(V),\plcp{\RT}{\rta{V}{1}} \}$ \KwTo $1$}{\label{alg:next:forloop}
			$r' \gets \ell - 1$\;
			$[\ell..r] \gets \maxi{\lcpinfty{\RT}}{y}{i}$\;
			$c \gets c + \rangecount{\ltarr{\RT}}{\ell}{r'}{i}{\sigma}$\;\label{alg:next:cntend}
		}
		\lIf{$c = 0$}{\label{alg:next:plcpstart}
			$p \gets -1$
		}
		\lElseIf{$\flmap{\RT}{c} > y$}{
			$p \gets 1$
		}
		\lElseIf{$\flmap{\RT}{c} = y$}{
			$p \gets \plcp{\RT}{\rta{V}{1}} - \pi(V) + 1$
		}
		\lElse{\label{alg:next:plcpend}
			$p \gets \min\{\plcp{\RT}{\rta{V}{1}}, \rnv{\lcpinfty{\RT}}{\flmap{\RT}{c} +1}{y}{-1}\} - \pi(V) + 1$
		}
		\lIf{$c = \rho$}{\label{alg:next:slcpstart}
			$s \gets -1$
		}
		\lElseIf{$\flmap{\RT}{c+1} < y + 1$}{
			$s \gets 1$
		}
		\lElseIf{$\flmap{\RT}{c+1} = y + 1$}{
			$s \gets \slcp{\RT}{\rta{V}{1}} - \pi(V) + 1$
		}
		\lElse{\label{alg:next:slcpend}
			$s \gets \min\{\slcp{\RT}{\rta{V}{1}}, \rnv{\lcpinfty{\RT}}{y +2}{\flmap{\RT}{c+1}}{-1}\} - \pi(V) + 1$
		}
		\KwRet $c, p, s$
		}
	\end{algorithm}
	
		\begin{lem}\label{lem:next}
		Let $\emptyset \neq \RT \subset \Sigma^+$, $\rho = \left|\carr{\RT}\right|$, and $V \in \extsigma^+$.
		Then Algorithm~\ref{alg:next} correctly computes $\cnt{\RT}{V}$, $\plcp{\RT}{V}$ and $\slcp{\RT}{V}$.
	\end{lem}
	\begin{proof}
		Algorithm~\ref{alg:next} takes $\pi(V)$, $y = \cnt{\RT}{\rta{V}{1}}$, $\plcp{\RT}{\rta{V}{1}}$, $\slcp{\RT}{\rta{V}{1}}$, and the cBWT index of $\RT$ as input.
		If $\pi(V) = \$$, then we return $\cnt{\RT}{V} = 0$, $\plcp{\RT}{V} = -1$ and $\slcp{\RT}{V} = 0$ in Line~\ref{alg:next:specialsymbol} since $\RT \subset \Sigma^+$.
		Thus, assume $\pi(V) > \$$.
		For each $i \in [1..\pi(V)]$, let $J_i = [\ell_i..r_i]$ maximal such that $r_{i} \leq y$ and $\lcpcount{V}{\conj{\RT}{\carr{\RT}[j]}} = i$ for each $j \in J_i$.
		The following statements are due to Lemma~\ref{lem:rotationomegaorder} and $\RT \subset \Sigma^+$.
		\begin{itemize}
			\item If $j \in [y+1..\rho]$, then $\conj{\RT}{\carr{\TT}[\lfmap{\RT}{j}]} \ctprec V$ if and only if $\ltarr{\RT}[j] \geq \pi(V)+1$ and $\lcpcount{\conj{\RT}{\carr{\TT}[i]}}{\rta{V}{j}} \geq \pi(V)+1$.
			\item If $j \in [r_{\pi(V)}+1..y]$, then $\conj{\RT}{\carr{\TT}[\lfmap{\RT}{j}]} \ctprec V$ if and only if $\ltarr{\RT}[j] \geq \pi(V)$.
			\item If $i \in [1..\pi(V)]$ and $j \in J_i$, then $\conj{\RT}{\carr{\TT}[\lfmap{\RT}{j}]} \ctprec V$ if and only if $\ltarr{\TT}[j] \geq i$.
		\end{itemize}
		We apply these results from Line~\ref{alg:next:cntstart} through Line~\ref{alg:next:cntend} to compute $\cnt{\RT}{V}$.
		Leveraging Lemma~\ref{lem:lcpinftyrota} and (the proof of) Lemma~\ref{lemma:lcpinftycomputation}, we then compute $\plcp{\RT}{V}$ and $\slcp{\RT}{V}$ from Line~\ref{alg:next:plcpstart} through Line~\ref{alg:next:plcpend} and from Line~\ref{alg:next:slcpstart} through Line~\ref{alg:next:slcpend}, respectively.
	\end{proof}

	\nextstatement*
	\begin{proof}
		We apply Algorithm~\ref{alg:next}.
		Correctness follows from \cref{lem:next}, and the time complexity is due to the loop of Algorithm~\ref{alg:next} in Line~\ref{alg:next:forloop} and Lemma~\ref{lemma:dynamicstringrmq}.
	\end{proof}

	\section{Properties of the Rotational Cartesian Tree Signature Encoding}\label{app:encodings}
	We show that the rotational Cartesian tree signature encoding is commutative with rotation.
	
	\begin{lem}\label{lem:rtacomm}
		For every $V \in \extsigma^+$ and $i,j \in [1..\left|V\right|]$T, $\rtsenc{\rta{V}{i}}[j] = \rta{\rtsenc{V}}{i}[j]$ .
	\end{lem}
	\begin{proof}
		Fix some $V \in \extsigma^+$ with $\lambda = \left|V\right|$ and $i,j \in [1..\lambda]$.
		It is easy to see that the statement holds if $\rta{V}{i}[j] = \$$.
		Thus, assume $\rta{V}{i}[j] \neq \$$.
		If $i+j \mod \lambda \neq 0$, then
		\begin{equation*}
			\begin{split}
				\rtsenc{\rta{V}{i}}[j] &= \rank{\Infty}{\mpdenc{\rta{\rta{V}{i}}{j}}}{\lambda}\\ 
				&- \rank{\Infty}{\mpdenc{\rta{V}{i}[j] \cdot \rta{\rta{V}{i}}{j}}[2..]}{\lambda} \\
				&= \rank{\Infty}{\mpdenc{\rta{V}{i+j}}}{\lambda}\\
				&- \rank{\Infty}{\mpdenc{V[i + j \mod \lambda] \cdot \rta{V}{i+j}}[2..]}{\lambda}\\
				&= \rtsenc{V}[i + j \mod \lambda] \\
				&= \rta{\rtsenc{V}}{i}[j].
			\end{split}
		\end{equation*}
		If $i+j \mod \lambda = 0$, then the statement follows similarly with $\rta{\rtsenc{V}}{i}[j] = \rtsenc{V}[\lambda]$ and $\rta{V}{i}[j] = V[\lambda]$.
	\end{proof}
	By Lemma~\ref{lem:rtacomm} and an examination of the definitions, we arrive at the following statement.
	\begin{cor}\label{cor:rtsencodingmatch}
		Let $V, U \in \extsigma^+$. 
		Then $\rtsenc{V} = \rtsenc{U}$ if and only if $V^2 \ctmatch U^2$.
	\end{cor}
	Then Lemma~\ref{lem:rotaencodingmatch} implies that the rotational modified parent distance encoding can be substituted for the rotational tree signature encoding in the definition of the $\omega$-preorder and the permutation $\textup{prev}_\TT$.
	The following statement is another consequence of Lemma~\ref{lem:rtacomm}.
	\begin{cor}\label{cor:definitionltarr}
		Let $\emptyset \neq \TT \subset \extsigma^+$ and $n$ the accumulated length of all texts in $\TT$.
		Then $\ftarr{\TT}[i] = \pi(\conj{\TT}{\carr{\TT}[i]})$ and $\ltarr{\TT}[i] = \rtsenc{\conj{\TT}{\carr{\TT}[i]}}[\left|\conj{\TT}{\carr{\TT}[i]}\right|]$ for each $i \in [1..n]$.
	\end{cor}
	Hence, we can define both $\ftarr{\TT}$ and $\ltarr{\TT}$ from the conjugate array and the $\rtsenc{..}$-encoded conjugates considered in $\TT$, and then define the LF-mapping with the statement of Corollary~\ref{cor:fastlfmapcompute}.

	\section{Locating Occurrences}\label{app:locate}
	
	In case we want to report the locations of occurrences, 
	we can construct $\carr{\TT}$ alongside the cBWT index and store it in $O(n + d\lg n)$ bits of space by means of sampling each $(\lg n)$-th text position~\cite[Sect.~3.2]{ferragina00fmindex}, where the latter term is due to the $d$ cycles of the LF-mapping corresponding to the $\rpdenc{..}$-encoded roots of the input texts, which require at least a sample each.
	We maintain an auxiliary bit string indicating the sampled locations in lex-order, and represent the sampled conjugate array, its auxiliary bit string, $\ftarr{\TT}$ and a copy of $\ltarr{\TT}$ by the data structure of the following result.
	\begin{lem}[\cite{munro15dynamic}]\label{lemma:dynamicstringsimple}
		A dynamic string of length $n$ over $[0..\sigma]$ with $\sigma \leq n^{O(1)}$ can be stored in a data structure occupying $O(n \lg \sigma)$ bits of space, supporting insertion, deletion and the queries \textup{access}, \textup{rank} and \textup{select} in $O(\frac{\lg n}{\lg\lg n})$ time.
	\end{lem}
	Then accessing the sampled conjugate array, checking if a position is sampled, and computing the LF- and FL-mapping takes $O(\frac{\lg n}{\lg\lg n})$ time.
	Consequently, we can answer locate queries in $O(m \frac{\lg \sigma \lg n}{\lg\lg n} + \mathit{occ} \frac{\lg^2 n}{\lg\lg n})$ time, where $m$ is the pattern length and $\mathit{occ}$ the number of occurrences.

	\section{Showcase Backward Search}\label{app:examplebackward}
	
	We give an example of the backward search for $P= \mathtt{375} \in \Sigma^+$ in our running example $\TT = \{\mathtt{512},\mathtt{5363},\mathtt{4478}\}$ leveraging the cBWT index of $\TT$ presented in Figure~\ref{fig:cbwt} and Figure~\ref{fig:cbwtfull}.
	
	Initially, $[\ell..r] = \crange{\TT}{P[4..]} = \crange{\TT}{\varepsilon} = [1..11]$, $h = \pi(P[3..]\cdot\$) = 0$, and $e = \rank{\Infty}{\mpdenc{P[3..]}}{3-3+1} = 1$, which means we enter Line~\ref{alg:crangeupd:2start} of Algorithm~\ref{alg:crangeupd} for the conjugate range update.
	We compute $c = \rangecount{\ltarr{\TT}}{1}{11}{0}{9} = 11$, $v = \rnv{\ltarr{\TT}}{1}{11}{-1} = 0$, $x = \select{0}{\ltarr{\TT}}{\rank{0}{\ltarr{\TT}}{11}} = 5$, $[\ell''..r''] = \maxi{\lcpinfty{\TT}}{5}{0+1} = [1..11]$, $y = \rangecount{\ltarr{\TT}}{1}{5-1}{0}{9} + \rangecount{\ltarr{\TT}}{5+1}{11}{0}{9} =  4 + 6 = 10$, $r' = \lfmap{\TT}{5} + 11 - (10 +1) = 11 + 11 - 11 = 11$, and $\ell' = 11 - 11 + 1 = 1$.
	Thus, we obtain $[\ell..r] = \crange{\TT}{P[3..]} = [1..11]$.
	
	Next, we compute $h = \pi(P[2..]\cdot\$) = 0$ and $e = \rank{\Infty}{\mpdenc{P[2..]}}{3-2+1} = 2$, which means we enter Line~\ref{alg:crangeupd:1start} of Algorithm~\ref{alg:crangeupd} for the conjugate range update.
	We compute $c = \rangecount{\ltarr{\TT}}{1}{11}{0}{0} = 4$, $r' = \lfmap{\TT}{\select{0}{\ltarr{\TT}}{11}} = \lfmap{\TT}{5} = 11$, and $\ell' = 11 - 4 + 1 = 8$, i.e., $[\ell..r] = \crange{\TT}{P[2..]} = [8..11]$.
	
	Last, we compute $h = \pi(P[1..]\cdot\$) = 2$ and $e = \rank{\Infty}{\mpdenc{P[1..]}}{3-3+1} = 1$, which means we enter Line~\ref{alg:crangeupd:2start} of Algorithm~\ref{alg:crangeupd} for the conjugate range update.
	We compute $c = \rangecount{\ltarr{\TT}}{8}{11}{2}{9} = 2$, $v = \rnv{\ltarr{\TT}}{8}{11}{1} = 2$, $x = \select{2}{\ltarr{\TT}}{\rank{2}{\ltarr{\TT}}{11}} = 11$, $[\ell''..r''] = \maxi{\lcpinfty{\TT}}{11}{3} = [11..11]$, $y = \rangecount{\ltarr{\TT}}{8}{10}{2}{9} + \rangecount{\ltarr{\TT}}{12}{11}{2}{9} = 1 + 0 = 1$, $r' = \lfmap{\TT}{11} + 2 - (1 + 1) = 5$, and $\ell' = 5 -2 + 1 = 4$, i.e., $\crange{\TT}{P} = \crange{\TT}{P[1..]} = [4..5]$.
	
	\section{Showcase cBWT Index Extension}\label{app:exampleextension}
	
	We extend the cBWT index of our running example $\TT$ by $S = \mathtt{73152}$ to index $\ST = \TT \cup \{S\}$.
	We have $z = \max\{\left|V\right| \mid V \in \ST\} = \left|S\right| = 5$ and $\crange{\TT}{S^3} = \emptyset$.
	For the extension, we follow the proof of Lemma~\ref{lem:cBWTextend}.
	The intermediate and final results are shown in Figure~\ref{figure:extension}.
	\begin{figure}
		\begin{subfigure}{\textwidth}
			\centering
			\begin{tabular}{|r||c|c|c|c|} 
				\hline
				$i$ & $\pi(\rta{V}{i})$ & $\cnt{\TT}{\rta{V}{i}}$ & $\plcp{\TT}{\rta{V}{i}}$ & $\slcp{\TT}{\rta{V}{i}}$\\
				\hline
				20 & $\mathtt{\$}$ & $\mathtt{0}$ & $\mathtt{-1}$ & $\mathtt{0}$ \\
				19 & $\mathtt{0}$ & $\mathtt{0}$ & $\mathtt{-1}$ & $\mathtt{1}$ \\
				18 & $\mathtt{0}$ & $\mathtt{7}$ & $\mathtt{1}$ & $\mathtt{2}$ \\
				17 & $\mathtt{2}$ & $\mathtt{3}$ & $\mathtt{1}$ & $\mathtt{1}$ \\
				16 & $\mathtt{0}$ & $\mathtt{9}$ & $\mathtt{2}$ & $\mathtt{2}$ \\
				15 & $\mathtt{0}$ & $\mathtt{11}$ & $\mathtt{2}$ & $\mathtt{-1}$ \\
				14 & $\mathtt{2}$ & $\mathtt{5}$ & $\mathtt{1}$ & $\mathtt{1}$ \\
				13 & $\mathtt{0}$ & $\mathtt{11}$ & $\mathtt{2}$ & $\mathtt{-1}$ \\
				12 & $\mathtt{3}$ & $\mathtt{5}$ & $\mathtt{1}$ & $\mathtt{1}$ \\
				11 & $\mathtt{0}$ & $\mathtt{11}$ & $\mathtt{2}$ & $\mathtt{-1}$ \\
				10 & $\mathtt{0}$ & $\mathtt{11}$ & $\mathtt{2}$ & $\mathtt{-1}$ \\
				9 & $\mathtt{2}$ & $\mathtt{5}$ & $\mathtt{1}$ & $\mathtt{1}$ \\
				8 & $\mathtt{0}$ & $\mathtt{11}$ & $\mathtt{2}$ & $\mathtt{-1}$ \\
				7 & $\mathtt{3}$ & $\mathtt{5}$ & $\mathtt{1}$ & $\mathtt{1}$ \\
				6 & $\mathtt{0}$ & $\mathtt{11}$ & $\mathtt{2}$ & $\mathtt{-1}$ \\
				\hline
				5 & $\mathtt{0}$ & $\mathtt{11}$ & $\mathtt{2}$ & $\mathtt{-1}$ \\
				4 & $\mathtt{2}$ & $\mathtt{5}$ & $\mathtt{1}$ & $\mathtt{1}$ \\
				3 & $\mathtt{0}$ & $\mathtt{11}$ & $\mathtt{2}$ & $\mathtt{-1}$ \\
				2 & $\mathtt{3}$ & $\mathtt{5}$ & $\mathtt{1}$ & $\mathtt{1}$ \\
				1 & $\mathtt{0}$ & $\mathtt{11}$ & $\mathtt{2}$ & $\mathtt{-1}$ \\
				\hline
			\end{tabular}
			\caption{Helper values for $V = S^4\$$.}
			\vspace*{2mm}
		\end{subfigure}
	
		\begin{subfigure}{\textwidth}
			\centering
			\begin{tabular}{|r||c|c|c|c|c|} 
				\hline
				$i$ & $\ftarr{\{S\}}[i]$ & $\ltarr{\{S\}}[i]$ &$\lcpinfty{\{S\}}[i]$  & $\plcp{\TT}{\conj{\{S\}}{\carr{\{S\}}[i]}}$ & $\slcp{\TT}{\conj{\{S\}}{\carr{\{S\}}[i]}}$\\
				\hline
				1 & $\mathtt{3}$ & $\mathtt{0}$ & $\mathtt{0}$ & $\mathtt{1}$ & $\mathtt{1}$\\
				2 & $\mathtt{2}$ & $\mathtt{0}$ & $\mathtt{1}$ & $\mathtt{1}$ & $\mathtt{1}$\\
				3 & $\mathtt{0}$ & $\mathtt{0}$ & $\mathtt{1}$ & $\mathtt{2}$ & $\mathtt{-1}$\\
				4 & $\mathtt{0}$ & $\mathtt{3}$ & $\mathtt{2}$ & $\mathtt{2}$ & $\mathtt{-1}$\\
				5 & $\mathtt{0}$ & $\mathtt{2}$ & $\mathtt{2}$ & $\mathtt{2}$ & $\mathtt{-1}$\\
				\hline
			\end{tabular}
			\label{test}
			\caption{The cBWT index of $\{S\}$ and helper values for $S$ in lex-order.}
			\vspace*{2mm}
		\end{subfigure}
		\begin{subfigure}{\textwidth}
			\centering
			\begin{tabular}{|r||c|c|c|c|} 
				\hline
				$i$ & $\ftarr{\ST}[i]$ & $\ltarr{\ST}[i]$ &$\lcpinfty{\ST}[i]$ & $E[i]$\\
				\hline
				1 & $\mathtt{1}$ & $\mathtt{0}$ & $\mathtt{0}$ & $\mathtt{0}$\\
				2 & $\mathtt{2}$ & $\mathtt{1}$ & $\mathtt{1}$ & $\mathtt{0}$\\
				3 & $\mathtt{2}$ & $\mathtt{0}$ & $\mathtt{1}$ & $\mathtt{0}$\\
				4 & $\mathtt{2}$ & $\mathtt{0}$ & $\mathtt{1}$ & $\mathtt{0}$\\
				5 & $\mathtt{2}$ & $\mathtt{0}$ & $\mathtt{1}$ & $\mathtt{0}$\\
				6 & $\mathtt{3}$ & $\mathtt{0}$ & $\mathtt{1}$ & $\mathtt{1}$\\
				7 & $\mathtt{2}$ & $\mathtt{0}$ & $\mathtt{1}$ & $\mathtt{1}$\\
				8 & $\mathtt{1}$ & $\mathtt{2}$ & $\mathtt{1}$ & $\mathtt{0}$\\
				\hline
			\end{tabular}
			\begin{tabular}{|r||c|c|c|c|} 
				\hline
				$i$ & $\ftarr{\ST}[i]$ & $\ltarr{\ST}[i]$ &$\lcpinfty{\ST}[i]$ & $E[i]$\\
				\hline
				9 & $\mathtt{1}$ & $\mathtt{2}$ & $\mathtt{1}$ & $\mathtt{0}$\\
				10 & $\mathtt{0}$ & $\mathtt{1}$ & $\mathtt{1}$ & $\mathtt{0}$\\
				11 & $\mathtt{0}$ & $\mathtt{1}$ & $\mathtt{2}$ & $\mathtt{0}$\\
				12 & $\mathtt{0}$ & $\mathtt{2}$ & $\mathtt{2}$ & $\mathtt{0}$\\
				13 & $\mathtt{0}$ & $\mathtt{2}$ & $\mathtt{2}$ & $\mathtt{0}$\\
				14 & $\mathtt{0}$ & $\mathtt{0}$ & $\mathtt{2}$ & $\mathtt{1}$\\
				15 & $\mathtt{0}$ & $\mathtt{3}$ & $\mathtt{2}$ & $\mathtt{1}$\\
				16 & $\mathtt{0}$ & $\mathtt{2}$ & $\mathtt{2}$ & $\mathtt{1}$\\
				\hline
			\end{tabular}
			\caption{The augmented cBWT index of $\ST$ before zeroing $E$.}
			\vspace*{2mm}
		\end{subfigure}
		\caption{Extending the cBWT index of the running example $\TT$ presented in Figure~\ref{fig:cbwt} and Figure~\ref{fig:cbwtfull} to index $\ST = \TT \cup \{S\}$, where $S = \mathtt{73152}$.}
		\label{figure:extension}
	\end{figure}

	\begin{landscape}
		\section{Full cBWT Example}\label{app:fullexample}
		\begin{figure}[h!]
		\centering
		\begin{tabular}{|c||l|l|c||c|l|c|c|l|c|c|c|l|} 
			\hline
			$i$ & $\conj{\TT}{i}$ & $\mpdenc{\conj{\TT}{i}^\omega[..12]}$ & $\icarr{\TT}[i]$ & $\carr{\TT}[i]$ & $\conj{\TT}{\carr{\TT}[i]}$ & $\flarr{\TT}[i]$ & $\ftarr{\TT}[i]$ & $\rtsenc{\conj{\TT}{\carr{\TT}[i]}}$ & $\ltarr{\TT}[i]$ & $\lfarr{\TT}[i]$ & $\lcpinfty{\TT}[i]$ & $\mpdenc{\conj{\TT}{\carr{\TT}[i]}^\omega[..12]}$\\
			\hline
			1 & $\mathtt{512}$ & $\mathtt{\Infty\Infty1131131131}$ & $\mathtt{9}$ & $\mathtt{8}$ & $\mathtt{4478}$ & $\mathtt{2}$ & $\mathtt{1}$ & $\mathtt{1210}$ & $\mathtt{0}$ & $\mathtt{8}$ & $\mathtt{0}$ & $\mathtt{\Infty11131113111}$ \\
			2 & $\mathtt{125}$ & $\mathtt{\Infty11311311311}$ & $\mathtt{3}$ & $\mathtt{9}$ & $\mathtt{4784}$ & $\mathtt{6}$ & $\mathtt{2}$ & $\mathtt{2101}$ & $\mathtt{1}$ & $\mathtt{1}$ & $\mathtt{1}$ & $\mathtt{\Infty11311131113}$ \\
			3 & $\mathtt{251}$ & $\mathtt{\Infty1\Infty113113113}$ & $\mathtt{7}$ & $\mathtt{2}$ & $\mathtt{125}$ & $\mathtt{7}$ & $\mathtt{2}$ & $\mathtt{210}$ & $\mathtt{0}$ & $\mathtt{9}$ & $\mathtt{1}$ & $\mathtt{\Infty11311311311}$ \\
			4 & $\mathtt{5363}$ & $\mathtt{\Infty\Infty1212121212}$ & $\mathtt{10}$ & $\mathtt{5}$ & $\mathtt{3635}$ & $\mathtt{10}$ & $\mathtt{2}$ & $\mathtt{2020}$ & $\mathtt{0}$ & $\mathtt{10}$ & $\mathtt{1}$ & $\mathtt{\Infty12121212121}$ \\
			5 & $\mathtt{3635}$ & $\mathtt{\Infty12121212121}$ & $\mathtt{4}$ & $\mathtt{7}$ & $\mathtt{3536}$ & $\mathtt{11}$ & $\mathtt{2}$ & $\mathtt{2020}$ & $\mathtt{0}$ & $\mathtt{11}$ & $\mathtt{1}$ & $\mathtt{\Infty12121212121}$ \\
			6 & $\mathtt{6353}$ & $\mathtt{\Infty\Infty1212121212}$ & $\mathtt{11}$ & $\mathtt{10}$ & $\mathtt{7844}$ & $\mathtt{8}$ & $\mathtt{1}$ & $\mathtt{1012}$ & $\mathtt{2}$ & $\mathtt{2}$ & $\mathtt{1}$ & $\mathtt{\Infty1\Infty111311131}$ \\
			7 & $\mathtt{3536}$ & $\mathtt{\Infty12121212121}$ & $\mathtt{5}$ & $\mathtt{3}$ & $\mathtt{251}$ & $\mathtt{9}$ & $\mathtt{1}$ & $\mathtt{102}$ & $\mathtt{2}$ & $\mathtt{3}$ & $\mathtt{1}$ & $\mathtt{\Infty1\Infty113113113}$ \\
			8 & $\mathtt{4478}$ & $\mathtt{\Infty11131113111}$ & $\mathtt{1}$ & $\mathtt{11}$ & $\mathtt{8447}$ & $\mathtt{1}$ & $\mathtt{0}$ & $\mathtt{0121}$ & $\mathtt{1}$ & $\mathtt{6}$ & $\mathtt{1}$ & $\mathtt{\Infty\Infty1113111311}$ \\
			9 & $\mathtt{4784}$ & $\mathtt{\Infty11311131113}$ & $\mathtt{2}$ & $\mathtt{1}$ & $\mathtt{512}$ & $\mathtt{3}$ & $\mathtt{0}$ & $\mathtt{021}$ & $\mathtt{1}$ & $\mathtt{7}$ & $\mathtt{2}$ & $\mathtt{\Infty\Infty1131131131}$ \\
			10 & $\mathtt{7844}$ & $\mathtt{\Infty1\Infty111311131}$ & $\mathtt{6}$ & $\mathtt{4}$ & $\mathtt{5363}$ & $\mathtt{4}$ & $\mathtt{0}$ & $\mathtt{0202}$ & $\mathtt{2}$ & $\mathtt{4}$ & $\mathtt{2}$ & $\mathtt{\Infty\Infty1212121212}$ \\
			11 & $\mathtt{8447}$ & $\mathtt{\Infty1\Infty111311131}$ & $\mathtt{8}$ & $\mathtt{6}$ & $\mathtt{6353}$ & $\mathtt{5}$ & $\mathtt{0}$ & $\mathtt{0202}$ & $\mathtt{2}$ & $\mathtt{5}$ & $\mathtt{2}$ & $\mathtt{\Infty\Infty1212121212}$ \\

			\hline
		\end{tabular}
		\caption{Data structures of our running example $\TT = \{\mathtt{512},\mathtt{5363},\mathtt{4478}\}$. 
		}
		\label{fig:cbwtfull}
	\end{figure}
	\end{landscape}
\end{document}